\documentclass[11pt]{article}
\usepackage{amsmath,amsthm,amssymb,color,verbatim,graphicx}
\newcommand{\remove}[1]{}
\newtheorem{thm}{Theorem}[section]
\newtheorem{claim}[thm]{Claim}
\newtheorem{lemma}[thm]{Lemma}
\newtheorem{define}[thm]{Definition}
\newtheorem{cor}[thm]{Corollary}

\renewcommand{\remove}[1]{}

\newcommand{\eps}{{\varepsilon}}
\renewcommand{\l}{\left}
\renewcommand{\r}{\right}

\newcommand{\x}{{\mathbf{x}}}
\newcommand{\de}{{\delta}}

\newcommand{\comments}[1]{}
\newcommand{\rank}{\textnormal{rank}}

\renewcommand{\deg}{\textnormal{deg}}
\renewcommand{\dim}{\textnormal{dim}}

\newcommand{\dist}{\textnormal{dist}}
\newcommand{\bias}{\textnormal{bias}}
\newcommand{\calF}{\mathcal{F}}

\newcommand{\spana}{\textnormal{span}}
\newcommand{\RM}{\textnormal{RM}}

\def\F{{\mathbb{F}}}

\newcommand{\bk}{\mathbf{k}}
\newcommand{\cF}{\mathcal{F}}
\newcommand{\cH}{\mathcal{H}}
\newcommand{\CC}{\mathcal{C}}
\newcommand{\R}{\mathbb{R}}
\newcommand{\N}{\mathbb{N}}

\newcommand{\E}{\mathbb{E}}

\newcommand{\K}{\mathbb{K}}
\newcommand{\Tr}{\mathbf{Tr}}

\renewcommand{\P}{\mathcal{P}}

\newcommand{\calH}{\mathcal{H}}

\newcommand{\calB}{\mathcal{B}}

\newcommand{\B}{\mathcal{B}}

\newcommand{\C}{\mathbb{C}}

\renewcommand{\Pr}{\mathbf{Pr}}

\newcommand{\ip}[1]{\left \langle #1 \right \rangle}

\def\draft{0}   

\ifnum\draft=1 
    \def\ShowAuthNotes{1}
\else
    \def\ShowAuthNotes{0}
\fi

\ifnum\ShowAuthNotes=0
\newcommand{\authnote}[2]{{ \footnotesize \bf{\color{red}[#1's Note: {\color{blue}#2}]}}}
\else
\newcommand{\authnote}[2]{}
\fi

\newcommand{\AuthornoteA}[2]{{\sf\small\color{red}{[#1: #2]}}}

\newcommand{\Anote}{\AuthornoteA{A}}

\begin{document}
\title{Bias vs structure of polynomials in large fields, and applications in information theory}

\author{
Abhishek Bhowmick\thanks{Research supported in part by NSF Grant CCF-1218723.}\\
Department of Computer Science\\
The University of Texas at Austin\\
\texttt{bhowmick@cs.utexas.edu}
\and
Shachar Lovett \thanks{Supported by NSF CAREER award 1350481}\\
Department of Computer Science and Engineering\\
University of California, San Diego\\
\texttt{slovett@ucsd.edu}}

\maketitle

\begin{abstract}
Let $f$ be a polynomial of degree $d$ in $n$ variables over a finite field $\F$. The polynomial is said to be unbiased if the distribution of $f(x)$ for a uniform input $x \in \F^n$ is close to the uniform distribution over $\F$, and is called biased otherwise. The polynomial is said to have low rank if it can be expressed as a composition of a few lower degree polynomials. Green and Tao [Contrib. Discrete Math 2009] and Kaufman and Lovett [FOCS 2008]
showed that bias implies low rank for fixed degree polynomials over fixed prime fields. This lies at the heart of many tools in higher order Fourier analysis. In this work, we extend this result to all prime fields (of size possibly growing with $n$). We also provide a generalization to nonprime fields in the large characteristic case. However, we state all our applications in the prime field setting for the sake of simplicity of presentation.

Using the above generalization to large fields as a starting point, we are also able to settle the list decoding radius of fixed degree Reed-Muller codes over growing fields. The case of fixed size fields was solved by Bhowmick and Lovett [STOC 2015], which resolved a conjecture of Gopalan-Klivans-Zuckerman [STOC 2008]. Here, we show that the list decoding radius is equal the minimum distance of the code for all fixed degrees, even when the field size is possibly growing with $n$.

Additionally, we effectively resolve the weight distribution problem for Reed-Muller codes of fixed degree over all fields, first raised in 1977 in the classic textbook by MacWilliams and Sloane [Research Problem 15.1 in Theory of Error Correcting Codes].
\end{abstract}

\section{Introduction}

Let $f$ be a polynomial of degree $d$ in $n$ variables over a finite field $\F$. The polynomial $f$ is said to be unbiased if the distribution of $f(x)$ for a uniform input $x \in \F^n$ is close to the uniform distribution over $\F$, and is called biased otherwise. We say that $f$ has low rank if it can be expressed as a composition of a few lower degree polynomials. The goal is to understand the structure of polynomials that are biased. Green and Tao~\cite{GT09} and Kaufman and Lovett~\cite{KL08} showed over fixed prime fields, that if a fixed degree polynomial is biased, then it has low rank. Such a result lies at the heart of many tools in higher order Fourier analysis. However, the bounds obtained from the above results have very weak dependence (Ackermann-type) on the field size $|\F|$ and the degree $d$, and thus are inefficient for large fields. In this work, we extend this to large fields, by proving bounds that are polynomial in the field size $|\F|$.

More precisely, we have the following. Let $\F$ be a prime finite field. Let $\P_d(\F^n)$ denote the family of polynomials $f:\F^n \to \F$ of total degree at most $d$. Let $e:\F \to \C$ be an additive character, $e(a)=\exp(2 \pi i a/|\F|)$.

\begin{thm}\label{thm:blr}Let $d,s\in \N$. Let $f \in \P_d(\F^n)$. Suppose that $|\E_{x \in \F^n}[e(f(x))]| \geq |\F|^{-s}$. Then, there exist $g_1,\ldots g_c \in \P_{d-1}(\F^n)$, $c=c^{(\ref{thm:blr})}(d,s)$, and $\Gamma:\F^c \rightarrow \F$, such that $f(x)=\Gamma(g_1(x),\ldots, g_c(x))$.
\end{thm}

Crucially, the rank $c$ is independent of both the field size $|\F|$ and the number of variables $n$. We show (Lemma~\ref{lem:complowdegree}) that $\Gamma$ itself is a low degree polynomial: if $\deg(g_i)=d_i$ then
$$
\Gamma(z_1,\ldots,z_c) = \sum_{e \in \N^c: \sum d_i e_i \le d} \alpha_e \prod_{i=1}^c z_i^{e_i}.
$$

\subsection{List Decoding Reed-Muller codes}
The notion of \emph{list decoding} was first introduced by Elias~\cite{Elias} and Wozencraft~\cite{Woz} to decode \emph{error correcting codes} beyond half the minimum distance. The objective of list decoding is to output all the codewords within a specified radius around the received word.
List decoding has applications in many areas of computer science including hardness amplification in complexity theory~\cite{STV, luca-xor}, construction of hard core predicates from one way functions~\cite{GL, AGS}, construction of extractors and pseudorandom generators~\cite{TZS, SU, Vad, GUV} and computational learning~\cite{KM,Jackson}. Despite so much progress, the largest radius up to which list decoding is tractable is still a fundamental open problem even for well studied codes like Reed-Solomon (univariate polynomials) and Reed-Muller codes (multivariate polynomials).

Reed-Muller codes (RM codes) were discovered by Muller in 1954. Fix a finite field $\F$ and $d,n \in \N$. The RM code $\RM_{\F}(n,d)$ is defined as follows. The message space consists of degree $d$ polynomials in $n$ variables over $\F$ and the codewords are evaluation of these polynomials on $\F^n$. The distance of two functions $f,g:\F^n \to \F$ is the fraction of points where they disagree,
$$
\dist(f,g):=\Pr_x[f(x) \ne g(x)]
$$
The minimal distance of a code is the minimal distance of any two distinct codewords. For $\RM_{\F}(n,d)$, this is well understood. When $d<|\F|$ the minimal distance is given by
$$
\dist_{\min}(\RM_{\F}(n,d))=1-\frac{d}{|\F|}.
$$
More generally, if $d=a(|\F|-1)+b$ for $0 \le b \le |\F|-1$ then the minimal distance is $|\F|^{-a} (1-\frac{b}{|\F|})$, but as we focus on large fields, we will always be in the regime of $d<|\F|$.

The list decoding radius of a code is the maximal radius, such that any ball of that radius (centered around an arbitrary function) contains only a few codewords. Let $\CC=\RM_{\F}(n,d)$. For $g:\F^n \to \F$, $0<\rho<1$ define
$$
B_{\CC}(g,\rho) := \l\{f \in \P_d(\F^n): \dist(f,g) \le \rho\r\}.
$$
and
$$
L_{\CC}(\rho) := \max_{g:\F^n \rightarrow \F} |B_{\CC}(g,\rho)|.
$$
The list decoding radius of $\CC$ is the maximal radius, up to which $L_{\CC}(\rho)$ is ``small".  In the regime of growing fields, ``small" is defined as as polynomial in the field size. It is easy
to see that the list decoding radius cannot exceed the minimal distance of the code. The Johnson bound~\cite{johnson} provides a general lower bound for the list decoding radius, which is determined
only by the minimal distance of the code. It is known to be tight in general, but it is conjecture not to be tight for special families of codes, for example Reed-Muller codes.

In the regime of constant size fields, it is known that the list decoding radius is in fact equal to the minimal distance of Reed-Muller codes. It was initially proved
by Goldreich and Levin~\cite{GL} and Goldreich, Rubinfield and Sudan~\cite{GRS} for linear polynomials, that is, $d=1$. Later, Gopalan, Klivans and Zuckerman~\cite{GKZ08} proved it for the binary field, $\F=\F_2$, and for general fixed prime fields $\F_p$ whenever $(p-1)|d$. They conjectured that is holds for all fixed $d,p$. Gopalan~\cite{Gopalan10} proved it for $d=2$. Bhowmick and Lovett~\cite{BL14} proved it for all fixed prime fields and all degrees. In this work, we extend this to all prime fields, with size possibly growing with $n$.

\begin{thm}\label{THM:listdecode}
Let $d,s \in \N$. There exists $c=c(d,s)$ such that the following holds.
For any prime finite field $\F$ with $|\F|>d$ and any $n \in \N$,
$$
L_{\RM_{\F}(n,d)}\l(1-\frac{d}{|\F|}-\frac{1}{|\F|^s}\r) \le |\F|^c.
$$
Moreover, for any $1 \le e<d$,
$$
L_{\RM_{\F}(n,d)}\l(1-\frac{e}{|\F|}-\frac{1}{|\F|^s}\r) \le |\F|^{c \cdot n^{d-e}}.
$$
\end{thm}
If $|\F| \leq d$, then the result follows from~\cite{BL14}.
There have been few results that show list decodability beyond the Johnson radius~\cite{DGKS08, GKZ08}. This work shows that Reed-Muller codes of fixed degree are list decodable beyond the Johnson radius.

\subsection{Non-prime fields}

The main focus of this paper is prime fields. However, we show (Theorem~\ref{thm:blrnprime}) that Theorem~\ref{thm:blr} can be extended to non-prime fields,
as long as their characteristics exceeds the degree of the polynomial studied. All the other results in this paper extend to this case as well, as given Theorem~\ref{thm:blr},
their proof extend without requiring any change.

\subsection{Proof Overview}
We first present a proof overview for Theorem~\ref{thm:blr}. The proof is along the lines of Green and Tao~\cite{GT09}. Let $f(x)$ be a polynomial of degree $d$ that is biased, that is $|\E_{x \in \F^n} e(f(x))| \ge |\F|^{-s}$. We first prove that there is a low rank approximation to the given polynomial $f$. That is, there exist $g_1,\ldots g_c \in \P_{d-1}(\F^n)$, $c=c(d,s,t)$, and $\Gamma:\F^c \rightarrow \F$, such that $$\Pr_{x \in \F^n}[f(x) \neq \Gamma(g_1(x),\ldots g_c(x))]\leq |\F|^{-t}.$$
In the regime of fixed finite fields, this was proved by Bogdanov and Viola~\cite{BV}, where the bound $c$ depends polynomially on the underlying parameters, including the error bound,
which means that it depends on the field size. Here, we obtain a variant of the lemma, where the bound is independent of the field size. This is crucial in the next step of the proof, where we show that if the error in approximation is small enough, then it can be converted to an exact computation, if we make the underlying polynomials ``random enough" by a regularization process. As this step increases the number
of polynomials tremendously, we cannot tolerate any dependence on the field size in the first part of the proof. The proof follows along the lines of~\cite{GT09} with appropriate modifications to tackle the case of growing field size.

The applications in effective algebraic geometry follow by using the principles of regularization, thereby reducing the dimension of the problem to a constant, solving it in constant dimension, and lifting the solution back to the original problem. They are typically straightforward applications of the former result.

The application to list decoding of Reed-Muller codes is more involved and uses the bias vs low rank theorem as one of the building blocks. Given a received function $g:\F^n \to \F$, the first step is to show that it is enough to bound the list size of a subcode of the Reed-Muller code, consisting of only the low rank polynomials. This step is similar to the work of Gopalan~\cite{Gopalan10}. We next show that the list decoding problem for low rank codes can be further reduced to the case where the center $g$ is of ``low complexity'', concretely, when $g$ is measurable with respect to a small polynomial factor of bounded degree. Unlike the case of fixed finite fields handled in~\cite{BL14}, we need to allow a number of potential low complexity centers for each received word. However, we show that this number is still polynomial in the field size, which allows to keep the number of codewords polynomial in the field size as well. Finally, we prove that the list size around such a low complexity center is bounded. The last part is similar to the analogous part in the previous work of the authors~\cite{BL14}.

\subsection{Related work}
Except for the work mentioned already, another related work which is worth mentioning is the recent algebraic regularity lemma by Tao~\cite{tao_algebraic}.
It improves upon the Szemer{\'e}di regularity lemma~\cite{Sze} in the setting where the graph is definable over a field of large characteristic.
In a high level, it shows that if a graph has the vertex set $\F^n$, for fixed $n$ and large $\F$, and edges defined by polynomial equalities of bounded complexity (fixed degree polynomials, fixed number
of variables, fixed number of logical operations) then the graph can be partitioned to a bounded number of subsets, such that all pairs are regular.
This should be compared to the Szemer{\'e}di regularity lemma, which can only guarantee this for most pairs.

A previous version of the paper had incorrectly claimed a certain application in effective algebraic geometry, which we have now removed. We acknowledge Guy Moshkovitz for pointing this out.

\subsection{Organization}

The rest of the paper is as follows. Section~\ref{sec:prelim} contains preliminaries. In Section~\ref{sec:approx} we show that any biased polynomial can be approximated by a composition of a small number of lower degree polynomials. In Section~\ref{sec:equal}, we show how to convert a good enough approximation to an exact computation. Section~\ref{sec:list} contains the application to list decoding of Reed-Muller codes.

\section{Preliminaries}\label{sec:prelim}
Let $\N$ denote the set of positive integers. For $n \in \N$, let $[n]:=\{1,2,\ldots , n\}$. We use $y=x \pm \eps$ to denote $y \in [x-\eps, x+\eps]$.
For $n \in \N$, and $x,y \in \C^n$, let $\langle x,y \rangle:=\sum_{i=1}^n x_i\overline{y_i}$ where $\overline{a}$ is the conjugate of $a$. Let $||x||_2:=\sqrt{\langle x,x \rangle}$.

Fix a prime field $\F=\F_p$. Let $|\cdot|: \F \to \{0,\ldots,p-1\} \subset \N$ be the natural map.
Let $e:\F \rightarrow \C$ be an additive character, defined as $e(a):=e^{2\pi i a / p}$. Recall that we denote by $\P_d(\F^n)$ the family of polynomials $f:\F^n \to \F$ of total degree at most $d$.
Given a function $f:\F^n \to \F$, its directional derivative in direction $h \in \F^n$ is $D_h f:\F^n \to \F$, given by
$D_h f(x) = f(x+h) - f(x)$.
Observe that if $f \in \P_d(\F^n)$ then $D_h f \in \P_{d-1}(\F^n)$ for all $h \in \F^n$. For $y_1,\ldots,y_m \in \F^n$ defined the iterative derivative as
$D_{y_1,\ldots,y_m} f = D_{y_1} \ldots D_{y_m} f$.
In particular, if $f \in \P_d(\F^n)$ and $m>d$ then $D_{y_1,\ldots,y_m} f=0$.

Let $X,Y$ be finite sets. Define $\Delta(Y):=\{q:Y \rightarrow \R_{\geq 0}  :\sum_{y \in Y} q(y)=1\}$ to be the probability simplex on $Y$. We embed $Y \subset \Delta(Y)$ in the obvious way: $y \in Y$ is mapped
to a unit vector $e_y$ with $1$ in coordinate $y$ and $0$ in all other coordinates. For a function $f:X \to Y$ let $p(f):X \to \Delta(Y)$ denote its corresponding embedding, given by $p(f)(x)=e_{f(x)}$.
Note that $\Delta(Y)$ is endowed with an inner product, as a subset of $\R^Y$. So, if $f,g:X \to Y$ then
$$
\Pr_{x \in \F^n}[f(x)=g(x)] = \E_{x \in \F^n} [\ip{p(f)(x), p(g)(x)}].
$$

\section{Bias implies low rank approximation}\label{sec:approx}

\begin{lemma}\label{lem:blr}Let $d,s,t \in \N$. Let $f \in \P_d(\F^n)$. Suppose $\l|\E_{x \in \F^n}[e(f(x))]\r| \geq |\F|^{-s}$. Then, there exist $g_1,\ldots g_c \in \P_{d-1}(\F^n)$, $c=c(d,s,t) = \binom{d+t+2s+3}{d}$, and $\Gamma:\F^c \rightarrow \F$, such that $$\Pr_{x \in \F^n}[f(x) \neq \Gamma(g_1(x),\ldots g_c(x))]\leq |\F|^{-t}.$$
Moreover, each $g_i$ is obtained as a derivative of $f$, $g_i = D_{h_i} f$ for some $h_i \in \F^n$.
\end{lemma}

We prove lemma~\ref{lem:blr} in this section. So, fix $f \in \P_d(\F^n)$ and let $\mu=\E_{x \in \F^n} \l[e(f(x))\r]$, where we assume $|\mu| \ge |\F|^{-s}$. We begin with the following simple claim.
\begin{claim}\label{claim:der}
For all $x \in \F^n$,
$$
\mu \cdot e(-f(x))=\E_{y \in \F^n} \l[e(D_yf(x))\r].
$$
\end{claim}

\begin{proof}
$\E_{y \in \F^n} \l[e(D_yf(x))\r]=\E_{y \in \F^n} \l[e(f(x+y))e(-f(x))\r]=\E_{y \in \F^n} \l[e(f(y))\r] \cdot e(-f(x))=\mu \cdot e(-f(x))$.
\end{proof}

Fix $x \in \F^n$. Pick $z=(z_1, \ldots z_k) \in (\F^n)^k$ uniformly for some $k$ to be specified later. For $a \in \F^k, z \in (\F^n)^k$, we shorthand $a \cdot z = \sum_{i=1}^k a_iz_i \in \F^n$. For $a \in \F^k\setminus \{0\}$, let $W_a(z)$ be the random variable (over the choice of $z$) defined as $$W_a(z):=e(D_{a \cdot z}f(x)).$$
For $a \neq 0$, we have $$\E_z[W_a(z)]=\E_y \l[e(D_y f(x))\r].$$ Also, observe that for distinct $\ell,m \in \F$, $$|e(\ell)-e(m)| \geq |\F|^{-1}.$$
We have the following.
\begin{claim}\label{claim:blrchebyshev}
If for $z \in (\F^n)^k$ it holds that
$$
\l|\frac{1}{|\F|^k-1}\sum_{a \neq 0}W_a(z)-\E_y \l[e(D_yf(x))\r]\r| \leq \frac{1}{2|\F|^{s+1}},
$$
then
$$
f(x)=\Gamma(D_{a\cdot z}f(x):a \in \F^k \setminus \{0\})
$$
where $\Gamma:\F^{|\F|^k-1}\rightarrow \F$ is some explicit function.
\end{claim}
\begin{proof}Since $|e(\ell)-e(m)| \geq |\F|^{-1}$ for $\ell \ne m$ and $|\mu|\geq |\F|^{-s}$, if we define $$\Gamma(y_1,\ldots y_{|\F|^k-1})=\arg \min_{\ell \in \F} \l|\frac{1}{|\F|^k-1}\sum_{i=1}^{|\F|^k-1}e(y_i)-e(-\ell)\mu\r|,$$
then by the assumption of the claim,
$$\Gamma(D_{a\cdot z}f(x):a \in \F^k \setminus \{0\})=f(x).$$
\end{proof}

Since the random variables $\{W_a(z): a \in \F^k \setminus \{0\}\}$ are pairwise independent, we have by Chebychev's inequality that
if we choose $k=t+2s+3$ then
\begin{equation}\label{eq:chebyshev1}
\Pr_{z \in (\F^n)^k}\l[\l|\frac{1}{|\F|^k-1}\sum_{a \neq 0}W_a(z)-\E_y \l[e(D_yf(x))\r]\r|\geq \frac{1}{2|\F|^{s+1}}\r]\leq \frac{4|\F|^{2s+2}}{|\F|^k-1} \leq \frac{1}{|\F|^t}.
\end{equation}
Thus, for all $x \in \F^n$, \
$$
\Pr_{z \in (\F^n)^k}[\Gamma(D_{a \cdot z}f(x):a \in \F^k \setminus \{0\})=f(x)]\geq 1-|\F|^{-t}.
$$
Therefore, by an averaging argument there exists $z \in (\F^n)^k$ for which
\begin{equation}
\Pr_{x \in \F^n}[\Gamma(D_{a \cdot z}f(x):a \in \F^k \setminus \{0\})=f(x)]\geq 1-|\F|^{-t}.
\end{equation}

%
%
%

We now prove our final claim, which shows that we only need a constant number of derivatives in order to approximate $f$ (instead of a number which is polynomial in $|\F|$).

\begin{claim}
Let $\calB=\{b \in \F^k: \sum_{j=1}^k |b_j|\leq d\}$. Then for any $a \in \F^k$,
$$
D_{a \cdot z} f(x) = \sum_{b \in \calB} \lambda_{a,b} D_{b \cdot z} f(x)
$$
for some $\lambda_{a,b} \in \F$.
\end{claim}

\begin{proof}
Let $|a|=\sum_{i=1}^k |a_i|$. We prove the claim by induction on $|a|$. If $|a| \le d$ the claim is straightforward, as assume $|a|>d$. As $f$ is a degree $d$ polynomial, we have for any $m>d$ and $y_1,\ldots,y_m \in \F^n$ that
$$
D_{y_1} \ldots D_{y_m} f \equiv 0.
$$
This translates to
$$
\sum_{c \in \{0,1\}^m} (-1)^{\sum c_i} f\l(x + \sum c_i y_i\r)=0.
$$
As the sum of the coefficients is zero, this implies that
$$
\sum_{c \in \{0,1\}^m} (-1)^{\sum c_i} D_{c \cdot y} f(x)=0.
$$
Apply this for $m=|a|$ and $y_1,\ldots,y_m$ set to $z_1$ repeated $a_1$ times, $z_2$ repeated $a_2$ times, up to $z_k$ repeated $a_k$ times. Then we obtain that
$$
\sum_{a' \le a} (-1)^{|a'|} D_{a' \cdot z} f(x)=0,
$$
where the sum is over all $a' \in \F^k$ such that $|a'_i| \le |a_i|$ for all $1 \le i \le k$. We conclude that $D_{a \cdot z}f (x)$ is a linear combination of $D_{a' \cdot z} f (x)$ for $a' \in \F^k$ with $|a'|<|a|$,
and apply the induction claim.
\end{proof}

This concludes the proof of Lemma~\ref{lem:blr}. We can approximate $f(x)$ correctly on $1-|\F|^{-t}$ fraction of the coordiantes, by a function of $|\B|\leq \binom{d+k}{d}$ polynomials of
lower degree, where $k=t+2s+3$.

\section{Bias implies low rank exact computation}\label{sec:equal}

The main theorem we prove is the following.

\noindent \textbf{Theorem~\ref{thm:blr}.} Let $d,s\in \N$. Let $f \in \P_d(\F^n)$. Suppose that $|\E_{x \in \F^n}[e(f(x))]| \geq |\F|^{-s}$. Then, there exist $g_1,\ldots g_c \in \P_{d-1}(\F^n)$, $c=c^{(\ref{thm:blr})}(d,s)$, and $\Gamma:\F^c \rightarrow \F$, such that $f(x)=\Gamma(g_1(x),\ldots g_c(x))$.
Moreover, each $g_i$ is obtained as a derivative of $f$, $g_i = D_{h_i} f$ for some $h_{i} \in \F^n$.

The proof is by induction on the degree $d$. The reader can verify that all the polynomials obtained throughout the proof are derivatives or iterated derivatives of $f$. Moreover,
iterated derivatives can be decomposed at the end back to few single derivatives, by applying the following identity iteratively:
$$
D_{h_1,h_2} f = D_{h_1+h_2} f - D_{h_1} f - D_{h_2} f
$$

But first, we define the notion of regularity followed by some important consequences of Theorem~\ref{thm:blr} which are required in the inductive proof of the same and might be of independent interest.

\subsection{Basic definitions}

\begin{define}[Rank] Let $d \in \N$ and $f:\F^n \rightarrow \F$. Then $\rank_d(f)$ is defined as the smallest integer $r$ such that there exist polynomials $h_1,\ldots , h_r:\F^n \rightarrow \F$ of degree $\leq d-1$ and a function $\Gamma:\F^r \rightarrow \F$ such that $f(x)=\Gamma(h_1(x),\ldots , h_r(x))$. If $d=1$, then the rank is $0$ if $f$ is a constant function and is
$\infty$ otherwise. If $f$ is a polynomial, then $\rank(f)=\rank_d(f)$ where $d=\deg(f)$.
\end{define}

\begin{define}[Factor] Let $X$ be a finite set. Then a factor $\B$ is a partition of the set $X$. The subsets in the partition are called atoms.
\end{define}

For finite sets $X$ and $Y$, recall that $\Delta(Y)$ is the probability simplex over $Y$, and that we embed $Y \subset \Delta(Y)$ and embed functions $f:X \to Y$ as functions $f:X \to \Delta(Y)$ in the obvious way.
For a factor $\B$ of $X$, a function $f:X\rightarrow \Delta(Y)$ is said to be measurable with respect to $\B$ if it is constant on the atoms of $\B$.
The average of $f$ over $\B$ is $\E[f|\B]:X \to \Delta(Y)$ defined as
$$
\E[f|\B](x)=\E_{y \in \B(x)}[f(y)]
$$
where $\B(x)$ is the atom containing $x$. Clearly, $\E[f|\B]$ is measurable with respect to $\B$.

A collection of functions $h_1,\ldots , h_c:X \rightarrow Y$ defines a factor $\B$ whose atoms are $\{x \in X:h_1(x)=y_1,\ldots,h_c(x)=y_c\}$ for every $(y_1,\ldots , y_c) \in Y^c$. We use $\B$ to also denote the map $x \mapsto \l(h_1(x),\ldots , h_c(x)\r)$. A function $f$ is measurable with respect to a collection of functions if it is measurable with respect to the factor the collection defines.

\begin{define}[Polynomial Factor] A polynomial factor $\B$ is a factor defined by a collection of polynomials $\calH=\{h_1,\ldots , h_c:\F^n \rightarrow \F\}$ and the factor is written as $\B_{\calH}$. The degree of the factor is the maximum degree of $h \in \calH$. With a slight abuse of notation, we would typically identify $\calH$ and $\B_{\calH}$.
\end{define}

Let $|\B|$ be the number of polynomials defining the factor. We define $||\B||:=|\F|^c$ to be the number of (possibly empty) atoms.

\begin{define}[Rank and Regularity of Polynomial Factor]Let $\B$ be a polynomial factor defined by $h_1,\ldots , h_c:\F^n \rightarrow \F$. Then, the rank of $\B$ is the least integer $r$ such that there exists $(a_1,\ldots , a_c) \in \F^c$, $\l(a_1, \ldots , a_c\r) \neq (0,\ldots , 0)$ for which the linear combination $h(x):=\sum_{i=1}^c  a_i h_i(x)$ has $\rank_d(h) \leq r$ where $d=\max_{i}\deg(a_i h_i)$. For a non decreasing function $r:\N \rightarrow \N$, a factor $\B$ is $r$-regular if its rank is at least $r(|\B|)$.
\end{define}

\begin{define}[Semantic and Syntactic refinement]Let $\B$ and $\B'$ be polynomial factors on $\F^n$. A factor $\B'$ is a syntactic refinement of $\B$, denoted by $\B' \succeq_{syn}\B$ if the set of polynomials defining $\B$ is a subset of the set of polynomials defining $\B'$. It is a semantic refinement, denoted by $\B' \succeq_{sem}\B$ if for every $x,y \in \F^n$, $\B'(x)=\B'(y)$ implies $\B(x)=\B(y)$.
\end{define}

\begin{lemma}[Polynomial Regularity Lemma]\label{lem:reg}  Let $r:\N \rightarrow \N$ be a non-decreasing function and $d \in \N$. Then there is a function $C_{r,d}^{(\ref{lem:reg})}:\N \rightarrow \N$ such that the following is true. Let $\B$ be a factor defined by polynomials $h_1,\dots , h_c:\F^n \rightarrow \F$ of degree at most $d$. Then, there is an $r$-regular factor $\B'$ defined by polynomials $h'_1,\ldots , h'_{c'}:\F^n \rightarrow \F$ of degree at most $d$ such that $\B' \succeq_{sem} \B$ and $c' \leq C_{r,d}^{(\ref{lem:reg})}(c)$.

Moreover if $\B \succeq_{syn} \hat{\B}$ for some polynomial factor $\hat{\B}$ that has rank at least $r(c')+c'+1$, then $\B' \succeq_{syn} \hat{B}$.
\end{lemma}
The proof of Lemma~\ref{lem:reg} is exactly along the lines of existing proofs in the literature, for example Lemma 2.3 in~\cite{GT09}, so we do not repeat it here.

For $(w_1,\ldots ,w_k), (w_1' , \ldots ,w_k') \in \F^k$, we write $(w_1,\ldots ,w_k) \prec (w_1', \ldots ,w_k')$ if $|w_i| \leq |w_i'|$ for all $i \in [k]$, where $|\cdot|$ is the canonical map from $\F$ to $\{0,1,\ldots ,p-1\}$.

\begin{define}[Affine system]\label{defn:affine}
An affine system is a set of linear forms $\{L_1,\ldots ,L_m\}$, where each $L_i:\F^k \to \F$ is defined by $L_i(x)=\sum_{j=1}^k  w_{i,j} x_j$, which satisfies the following:
\begin{itemize}
\item $w_{i,1}=1$ for all $i \in [m]$.
\item If $L'(x)=\sum_{j=1}^k  w'_j x_j$, where $w'_1 = 1$ and $w' \prec w_i$ for some $i \in [m]$, then $w'=w_j$ for some $j \in [m]$.
\end{itemize}
\end{define}

\subsection{Inverse Gowers norm for polynomial phases}
\begin{thm}\label{thm:invGow}
Suppose Theorem~\ref{thm:blr} is true up to order $d$. Let $d,s\in \N$. Let $f \in \P_d(\F^n)$. Suppose $||e(f)||_{U^d} \geq |\F|^{-s}$. Then, $\rank(f)\leq c^{(\ref{thm:invGow})}(d,s)$.
\end{thm}
\begin{proof}
We have
$$
\l|\E_{x,y_1,\ldots ,y_d \in \F^n}\l[e\l(D_{y_1,\ldots ,y_d}f(x)\r)\r]\r|=||e(f)||_{U^d}^{2^d} \geq |\F|^{-s}.
$$
Let $g:\F^{n(d+1)}\rightarrow \F$ be defined as $$g(x,y_1,\ldots ,y_d):=D_{y_1,\ldots ,y_d}f(x).$$
By Theorem~\ref{thm:blr},
$$
\rank(g)\leq c^{(\ref{thm:blr})}(d,s/2^d).$$
By Taylor's theorem, since we assume $d<|\F|$, $$f(x)=\frac{D_{x,\ldots ,x}f(0)}{d!}+h(x),$$ where $h \in \P_{d-1}(\F^n)$. Since, $g(0,x,\ldots ,x) \equiv D_{x,\ldots ,x}f(0)$, we conclude that $\rank(f) \leq
\rank(g)+1 \le c^{(\ref{thm:blr})}(d,s/2^d)+1$. Choosing $c^{(\ref{thm:invGow})}(d,s)$ large enough such that $c^{(\ref{thm:invGow})}(d,s)\geq c^{(\ref{thm:blr})}(d,s/2^d)+1$ finishes the proof.
\end{proof}

\subsection{Equidistribution of atoms}
The next lemma shows that a regular factor has atoms of roughly equal size.
\begin{lemma}[Size of atoms]\label{lem:atomsize}Suppose Theorem~\ref{thm:blr} is true up to order $d$.
Let $\B=\{h_1,\ldots,h_c\}$ be a polynomial factor of degree at most $d$.
Given $s \in \N$, assume that $\B$ has rank at least $c^{(\ref{thm:blr})}(d,s)$. Then for every $b \in \F^c$, $$ \Pr_{x \in \F^n}[\B(x)=b]=\frac{1}{||\B||} \pm \frac{1}{|\F|^s}.$$
\end{lemma}
\begin{proof}For any $b \in \F^c$,
\begin{eqnarray*}
\Pr[\B(x)=b]&=&\frac{1}{|\F|^c}\sum_{a \in \F^c}\E_x \l[e\l(\sum_i a_i(h_i(x)-b_i)\r)\r]\\
&=&\frac{1}{|\F|^c}\pm \frac{1}{|\F|^c}\sum_{0 \neq a  \in \F^c}\l|\E_x \l[e\l(\sum_i a_i h_i(x)\r)\r]\r|\\
&=&\frac{1}{|\F|^c} \pm \frac{1}{|\F|^s}
\end{eqnarray*}

The last line follows because of the following. Suppose for some $a \neq 0$, $\l|\E_x \l[e\l(\sum_i a_i h_i(x)\r)\r]\r|>\frac{1}{|\F|^s}$, then by Theorem~\ref{thm:blr}, $\rank(\sum_i a_i h_i) \leq c^{(\ref{thm:blr})}(d,s)$. This contradicts the assumption on the rank of $\B$.
\end{proof}

\subsection{Near orthogonality of affine linear forms}

\begin{lemma}[Near orthogonality]\label{lem:equiaffine}
Suppose Theorem~\ref{thm:blr} is true up to order $d$. Let $c,d,p,s,m,k  \in \N$.
Let $\B=\{h_1,\ldots,h_c\}$ be a polynomial factor of degree at most $d$. Assume $\B$ has rank at least $r^{(\ref{lem:equiaffine})}(d,k,s)$. Let $(L_1,\ldots , L_m)$ be an affine system on $k$ variables. Let $\Lambda=(\lambda_{ij})_{i \in [c], j \in [m]}$ be a tuple of integers. Define $$h_{\Lambda}(x_1,\ldots ,x_{k})=\sum_{i \in [c], j \in [m]}\lambda_{ij}h_i(L_j(x_1,\ldots , x_{k})).$$ Then one of the following is true.
\begin{enumerate}
\item
$h_{\Lambda}\equiv 0$. Moreover, for every $i \in [c]$, it holds that $\sum_{j=1}^m \lambda_{ij}g_i(L_j(\cdot))\equiv 0$ for all $g_i \in \P_d(\F^n)$.
\item $h_{\Lambda}\not \equiv 0$. Moreover, $\l|\E[e(h_{\Lambda}(x_1,\ldots , x_{k})]\r|\leq |\F|^{-s}$.
\end{enumerate}
\end{lemma}
Again, the proof is exactly along the lines of Theorem 3.3 in~\cite{BFHHL} taking care of the dependence on $|\F|$ now, followed by an application of Theorem~\ref{thm:invGow}.
As a corollary, we state the above result for the case of parallelepipeds. We will need this in the inductive proof of Theorem~\ref{thm:blr}.

\subsection{Equidistribution of parallelepipeds}
We first set up some definitions following Section 4 in~\cite{GT09}. Throughout this subsection, let $\B=\{h_1,\ldots,h_c\}$ be a polynomial factor of degree at most $d$. We assume $\B$ has rank at least $r^{(\ref{thm:invGow})}(d,s)$. For $i \in [d]$, $M_i$ denotes the number of polynomials in $\B$ of degree exactly equal to $i$. Let $\Sigma:=\otimes_{i \in [d]}\F^{M_i}$.

\begin{define}[Faces and lower faces]Let $k \in \N$ and $0 \leq k' \leq k$. A set $F \subseteq \{0,1\}^k$ is called a face of dimension $k'$ if $$F=\{b:b_i=\de_i, i\in I\},$$ where $I \subseteq [k]$, $|I|=k-k'$ and $\de_i \in \{0,1\}$. If $\de_i=0$ for all $i \in I$, the $F$ is a lower face. Thus, it is equivalent to the power set of $[k]\setminus I$.
\end{define}

\begin{define}[Face vectors and parallelepiped constraints]Let $i_0 \in [d]$, $j_0 \in [M_{i_0}]$ and $F \subseteq \{0,1\}^k$. Let $r(i_0,j_0,F) \in \Sigma^{\{0,1\}^k}$ indexed as $r(i,j,\omega)=(-1)^{|\omega|}$ if $i=i_0, j=j_0$ and $\omega \in F$ and zero otherwise. This is called a face vector. If $F$ is a lower face, then it corresponds to a lower face vector. If $\dim(F) \geq i_0+1$, then it is a relevant face (lower face) vector. A vector $(t(\omega):\omega \in \{0,1\}^k) \in \Sigma^{\{0,1\}^k}$ satisfies the parallelepiped constraints if it is orthogonal to all the relevant lower face vectors.
\end{define}

Let $\Sigma_0 \subseteq \Sigma^{\{0,1\}^k}$ be the subspace of vectors satisfying the parallelepiped constraints.

\begin{claim}[Dimension of $\Sigma_0$, Lemma 4.4~\cite{GT09}]\label{claim:dimSigma}Let $d<k$. Then, $$\dim(\Sigma_0)=\sum_{i=1}^d M_i\sum_{0 \leq j \leq i}\binom{k}{j}.$$
\end{claim}

\begin{lemma}[Equidistribution of parallelepipeds]\label{lem:parallel}Suppose Theorem~\ref{thm:blr} is true up to order $d$. Given $s,d<k \in \N$, let $\B$ be a polynomial factor of rank at least $c^{(\ref{lem:parallel})}(k,s)$ defined by polynomials $h_1,\ldots , h_c:\F^n \rightarrow \F$ of degree at most $d$. For every $t \in \Sigma_0$ and $x$ such that $\B(x)=t(0)$, $$ \Pr_{y_1,\ldots ,y_k}[\B(x+\omega \cdot y)=t(\omega) \ \forall \ \omega \in \{0,1\}^k]=\frac{1}{|\F|^{\sum_{i=1}^d M_i\sum_{1 \leq j \leq i}\binom{k}{j}}} \pm \frac{1}{|\F|^s}.$$
\end{lemma}
\begin{proof}This immediately follows from the dimension of $\Sigma_0$ (Claim~\ref{claim:dimSigma}) and Lemma~\ref{lem:equiaffine} applied to the parallelepiped.
\end{proof}

\subsection{Proof of Theorem~\ref{thm:blr}}The proof of Theorem~\ref{thm:blr} is by induction on $d$ and follows along the lines of Theorem 1.7 in~\cite{GT09}. We sketch the proof here.

\begin{proof}[Proof of Theorem~\ref{thm:blr}]
The base case of $d=1$ is trivial. Indeed, if a linear polynomial $f:\F^n \to \F$ satisfies $|\E[e(f(x)]| \geq |\F|^{-s}$, then by orthogonality of linear polynomials, we have $f(x)$ is a constant and hence has rank $0$.
Now, suppose the hypothesis is true for degree $d-1$. Let $t \in \N$ depending on $d$ be specified later. We have $|\E[e(f(x))]| \geq |\F|^{-s}$. By Lemma~\ref{lem:blr}, there exists $\B=\{h_1,\ldots h_c:h_i \in \P_{d-1}(\F^n)\}$, $c=c(d,s,t)$, and $\Gamma:\F^c \rightarrow \F$, such that $$\Pr[f(x) \neq \Gamma(h_1(x),\ldots h_c(x))]\leq |\F|^{-t}.$$ Let $r:\N \to \N$ be a growth function that depends on $d$ and will be specified later. Regularize $\B$ to an $r$-regular polynomial factor $\B'=\{h_1',\ldots ,h_{c'}'\}$, $c' \leq C^{(\ref{lem:reg})}_{r,d}(c)$. Thus, we have for an appropriate $\Gamma':\F^{c'} \to \F$ that
$$
\Pr[f(x) \neq \Gamma'(h'_1(x),\ldots h'_{c'}(x))]\leq |\F|^{-t}.
$$

In the rest of the proof, we prove that $f$ is $\B'$-measurable. This will finish the proof.
We will assume that $r(j) \geq c^{(\ref{lem:atomsize})}(d,2t+j)$ for all $j \in \N$. By Markov's inequality and Lemma~\ref{lem:atomsize}, for at least $1-|\F|^{-t/4}$ fraction of atoms $A$, $$\Pr_{x \in A}[f(x) \neq \Gamma'(h_1'(x),\ldots h_{c'}'(x))]\leq |\F|^{-t/4}.$$

The first step is to prove that on such atoms, $f$ is constant. Fix such an atom $A$ and let $A' \subseteq A$ be the set where $f(x)=\Gamma'(h_1'(x),\ldots h_{c'}'(x))$.
\begin{lemma}Let $t$ be large enough depending on $d$. Let $x \in A$ be arbitrary. Then there is an $h \in (\F^n)^{d+1}$ such that $x+\omega \cdot h \in A'$ for all $\omega \in \{0,1\}^{d+1} \setminus 0^{d+1}$.
\end{lemma}
The proof is exactly as in Lemma 5.2 in~\cite{GT09}. We omit it here. Continuing, since $f \in \P_d(\F^n)$, we have $$\sum_{\omega \in \{0,1\}^{d+1}}(-1)^{|\omega|}f(x+\omega \cdot h)=0.$$ Now, by the above lemma, we have $f(x+\omega \cdot h) \equiv c_A$ for $\omega \neq 0$, where $c_A$ is a constant that depends on $A$. Thus, $f(x) \equiv c_A$.

This finishes the first step. Thus, we have for $1-|\F|^{-t/4}$ fraction of the atoms $A$, call them good atoms, $f(x)=c_A$. The final step shows that for any arbitrary atom $A$, there are good atoms $A_{\omega}$, $0 \neq \omega \in \{0,1\}^{d+1}$ such that the vector $t=\B(A_{\omega}) \in \Sigma^{\{0,1\}^{d+1}}$ satisfies the parallelepiped constraints. It is enough to find one parallelepiped for which $x+\omega \cdot h$ lie in good atoms for $\omega \neq 0$. Indeed, let $x \in A$ be arbitrary. Pick $h_1,\ldots ,h_{d+1}$ randomly. The probability that for a fixed $\omega \neq 0$,  $x+\omega \cdot h$ lies in a good atom is at least $1-|\F|^{-t/4}>1-2^{-2d}$ for $t$ large enough. The result now follows by a union bound over $\omega \in \{0,1\}^{d+1}$.
\end{proof}

\subsection{Some more consequences}

\paragraph{Degree preserving lemma.}

\begin{lemma}[Degree Preserving Lemma]\label{lem:degree}Let $c,d,D \in \N$ with $d<|\F|$.
Let $\B=\{h_1,\ldots,h_c\}$ be a polynomial factor of degree at most $d$, and rank at least $r^{(\ref{lem:degree})}(c,d,D)$. For $\Gamma:\F^c \rightarrow \F$, let $F:\F^n \rightarrow \F$ be defined by $F(x)=\Gamma(h_1(x),\ldots , h_c(x))$. Let $\deg(F)=D$. Then, for every set of polynomials $h'_1,\ldots h'_c:\F^n \rightarrow \F$ with $\deg(h'_i) \leq \deg(h_i)$ for all $i \in [c]$, if $G:\F^n \rightarrow \F$ is defined by $G(x)=\Gamma(h'_1(x),\ldots , h'_c(x))$, we have $\deg(G) \leq D$.\end{lemma}
We omit the proof here as it can be readily adapted from Theorem 4.1 in~\cite{BFHHL}.

\paragraph{Faithful composition.}

\begin{lemma}[Faithful composition lemma]\label{lem:complowdegree}Let $c,d,D \in \N$.
Let $\B=\{h_1,\ldots,h_c\}$ be a polynomial factor of degree at most $d$, and rank at least
$r^{(\ref{lem:degree})}(c,d,D)$. Let $\Gamma:\F^c \rightarrow \F$ be defined by $\Gamma(z)=\sum_{s \in S}a_s\prod_{i}z_i^{s_i}$, for some $S \subset \N^c$
and where $a_s \ne 0$ for all $s \in S$. Define $F:\F^n \to \F$ by
$$
F(x) = \sum_{s \in S} a_s \prod_{i=1}^c h_i(x)^{s_i}.
$$
Assume that $\deg(F) = D$. Then for every $s \in S$, $$\sum_{i=1}^c s_i \cdot \deg(h_i) \leq D.$$
\end{lemma}

\begin{proof}
Let $d_i=\deg(h_i)$. Define new variables $x'=\{x'_{i,j}: i \in [c], j \in [d_i]\}$ with $x'_{i,j} \in \F$. Define new polynomials $h'_i(x')=\prod_{j=1}^{d_i}x'_{i,j}$, where we note that $h'_1,\ldots,h'_c$ are defined over
disjoint sets of variables, and that $\deg(h'_i)=\deg(h_i)$. Define $G(x')=\Gamma(h'_1(x'),\ldots ,h'_c(x'))$. Since $\B$ has rank at least $r^{(\ref{lem:degree})}(c,d,D)$, we have
by Lemma~\ref{lem:degree} that $\deg(G) \leq D$. Expanding the definition of $\Gamma$ we have
$$
G(x')=\sum_{s \in S} a_s \prod_{i=1}^c \prod_{j=1}^{d_i} (x'_{i,j})^{s_i}.
$$
Note that each $s \in S$ corresponds to a unique monomial of degree $\sum_{i=1}^c d_i s_i$, and the monomials cannot cancel each other. The lemma follows.
\end{proof}

\paragraph{Hyperplane Restriction.}
Next, we show that the notion of rank is robust to hyperplane restrictions. More precisely, we have the following.
\begin{lemma}\label{lem:hyperplanerank}
Let $f \in \P_d(\F^n)$ such that $\rank(P) \geq r$. Let $H$ be a hyperplane in $\F^n$. Then the restriction of $f$ to $H$ has rank at least $r-d-1$.
\end{lemma}

We note that the existing results prove a lower bound  of $r-|\F|$, but with a slight modification (which we show below) we are able to prove a lower bound of $r-d-1$.

\begin{proof}
Without loss of generality, let $H$ be defined by $x_1=0$. For $x \in \F^n$ let $x'=x_2 \ldots x_n \in \F^{n-1}$ so that $x=(x_1,x')$
and $f|_H(x')=f(0,x')$. Define $f_i:\F^{n-1} \to \F$ by
$$
f_i(x') = f(i,x') - f(0,x').
$$
Clearly, $f(x) = \Gamma(x_1, f|_H(x), f_1(x'), \ldots, f_{|\F|-1}(x'))$ for some explicit $\Gamma:\F^{|\F|+1} \to \F$. For $v_i=(i,0,\ldots,0) \in \F^n$, we have that $f_i$ is the restriction
of $D_{v_i} f$ to $H$, and hence $\deg(f_i) \le \deg(D_{v_i} f) \le d-1$. To conclude the proof, we show that for any $j>d$, $f_j(x')$ can be expressed a linear combination of $\{f_1(x'),\ldots,f_d(x')\}$. This will
imply that in fact, $f(x) = \Gamma'(x_1, f|_H, f_1(x'), \ldots, f_{d}(x'))$ for some $\Gamma':\F^{d+2} \to \F$ and, since $\deg(f_i)<\deg(f)$ for all $i$, will show that $\rank(f) \le \rank(f|_H)+d+1$.

To conclude the proof, fix $j>d$. We will show that $f_j(x')$ is a linear combination of $\{f_1(x'),\ldots,f_{j-1}(x')\}$, which by induction will show the claim.
As $\deg(f) \le d$ we have
$$
\underbrace{D_{x_1}D_{x_1}\ldots D_{x_1}}_{j \textrm{ times}}f(x)=0.
$$
Writing this explicitly, and restricting to $x=(0,x')$, we obtain that
$$
\sum_{i=0}^j (-1)^i {j \choose i} f(i,x')=0,
$$
which in turn implies that
$$
\sum_{i=1}^j (-1)^i {j \choose i} f_i(x')=0.
$$
Thus, $f_j(x')$ is a linear combination of $\{f_1(x'),\ldots,f_{j-1}(x')\}$, as claimed.
\end{proof}

\subsection{Algorithmic Aspects}
It is easy to see that the existential proof of the main theorem can be made algorithmic.
\begin{lemma}\label{lem:algoblr}Let $d,s\in \N$. There is a randomized algorithm that on input $f \in \P_d(\F^n)$ with $\l|\E[e(f(x))]\r| \geq |\F|^{-s}$, runs in time $O(|\F|^c \cdot n^d)$ and outputs $g_1,\ldots , g_c \in \P_{d-1}(\F^n)$, $c=c^{(\ref{thm:blr})}(d,s)$, and $\Gamma:\F^c \rightarrow \F$, such that $f(x)=\Gamma(g_1(x),\ldots g_c(x))$.
\end{lemma}

The proof of the above follows similar to Theorem 1.4 in~\cite{bht_soda}. It can be derandomized using either Viola's generator~\cite{viola2009sum} or Bogdanov's generator~\cite{bog05} for low degree polynomials. To use Bogdanov's generator, one requires the field size to be at least superlogarithmic in $n$. For details, see the proof of Theorem 1.2 in~\cite{arnab}.

\subsection{Extension to non-prime fields}\label{sec:nprime}
In this section, we extend Theorem~\ref{thm:blr} to non-prime fields in the large characteristic case. When the field size is fixed, this was shown in \cite{BhBh15}. We focus on the setting when the field size can grow with $n$. Let $\K:=\F_{p^m}$ with $m \in \N$ and $p$ prime, where we assume that $p>d$.
Let $\chi:\K \to \C$ be a nontrivial additive character, that is $\chi(a+b)=\chi(a) \chi(b)$ for all $a,b \in \K$, and $\chi \not \equiv 1$.

\begin{thm}\label{thm:blrnprime}Let $d<p,s\in \N$. Let $f \in \P_d(\K^n)$. Suppose that $|\E_{x \in \K^n}[\chi(f(x)]| \geq |\K|^{-s}$. Then, there exist $g_1,\ldots g_c \in \P_{d-1}(\K^n)$, $c=c^{(\ref{thm:blrnprime})}(d,s)$, and $\Gamma:\K^c \rightarrow \K$, such that $f(x)=\Gamma(g_1(x),\ldots g_c(x))$.
\end{thm}

We prove Theorem~\ref{thm:blrnprime} in this subsection.
The proof follows closely the proof for the prime field case, and we highlight the differences. First, we note that any additive character
$\chi:\K \to \C$ can be factored as $\chi(x)=e(a \Tr(x))$ where $\Tr:\K \to \F_p$ is the trace map and $e:\F_p \to \C$ is given as usual by $e(x)=\exp(2 \pi i x/p)$. A nontrivial character
corresponds to $a \ne 0$. We may assume without loss of generality that $a=1$ by replacing $f$ with $af$, since $a \Tr(f) = \Tr(a f)$ and since
if $af$ has low rank then so does $f$. So, from now on
we assume that $\chi(x)=e(\Tr(x))$.

We first show that $\Tr(f):\K^n \to \F_p$ can be well approximated by the traces of a few lower degree polynomials.

\begin{lemma}\label{lem:blrnprime_approx}Under the conditions of Theorem~\ref{thm:blrnprime}, there exist $g_1,\ldots g_c \in \P_{d-1}(\K^n)$, $c=c(d,s,t) = \binom{d+t+2s+3}{d}$, and $\Gamma:\F_p^c \rightarrow \F_p$, such that $$\Pr_{x \in \K^n}[\Tr(f(x)) \neq \Gamma(\Tr(g_1(x)),\ldots,\Tr(g_c(x)))]\leq |\K|^{-t}.$$
\end{lemma}

\begin{proof}
The proof is identical to the proof of Lemma~\ref{lem:blr}, applied to the function $f'=\Tr(f):\K^n \to \F_p$. The reader can verify that the proof of
Lemma~\ref{lem:blr} does not really require that $\K$ is a prime field. Instead, in only relies on three properties of $f'$, which are true in the prime
field case, but are also true in the current case. These are:
(i) $f'(x)$ has bias at least $|\K|^{-s}$; (ii) $f'(x)$ takes few possible values ($p$ in our case); and (iii)
$f'(x)$ is annihilated by $d+1$ derivatives. Conditions (i) and (ii) follow by our assumptions.
Condition (iii) is true since
$$
D_{y_1,\ldots,y_{d+1}} f' = D_{y_1,\ldots,y_{d+1}} \Tr(f) = \Tr( D_{y_1,\ldots,y_{d+1}} f) = \Tr(0)=0.
$$
The functions $g'_1,\ldots,g'_c$ obtained in Lemma~\ref{lem:blr} are
derivatives of $f'$, that is $g'_i = D_{h_i} f'$ for some $h_i \in \K^n$. However since $f'=\Tr(f)$ we have that
$g'_i=D_{h_i} \Tr(f) = \Tr(D_{h_i} f)$, hence we can take $g_i = D_{h_i} f \in \P_{d-1}(\K^n)$.
\end{proof}

We next show that this implies that $\Tr(f)$ has low rank. Moreover, the factors are all traces of low degree polynomials over $\K^n$.

\begin{lemma}\label{lem:blrnprime_trace} Under the conditions of Theorem~\ref{thm:blrnprime}, there exist $g_1,\ldots g_c \in \P_{d-1}(\K^n)$, $c=c^{(\ref{lem:blrnprime_trace})}(d,s)$, and $\Gamma:\F_p^c \rightarrow \F_p$, such that $\Tr(f)(x)=\Gamma(\Tr(g_1(x)),\ldots, \Tr(g_c(x)))$.
\end{lemma}

\begin{proof}
Identify $\K \cong \F_p^m$. Under this identification, we can view $\Tr(f)$ as a polynomial in $\P_d(\F^{mn})$. By Lemma~\ref{lem:blrnprime_approx},
it can be well approximated by a function of $\Tr(g_1),\ldots,\Tr(g_c) \in \P_{d-1}(\F^{mn})$. Hence, by Theorem~\ref{thm:blr}, $\Tr(f)$ has low rank. Moreover,
all the polynomials obtained in the proof of Theorem~\ref{thm:blr} are derivatives of $\Tr(f),\Tr(g_1),\ldots,\Tr(g_c)$. But as we already observed, derivatives and traces commute,
hence all the polynomials in the factorization obtained in Theorem~\ref{thm:blr} are traces of derivatives of $f$, as claimed.
\end{proof}

To conclude, we show that since $\Tr(f)$ can be factored as a function of a few traces of low degree polynomials, then $f$ must have a low rank. To this end, define
$\Delta f \in \P_d(\K^{nd})$ to be the derivative polynomial of $f$, that is
$$
\Delta f(y_1,\ldots,y_d) = D_{y_1,\ldots,y_d} f (x) = D_{y_1,\ldots,y_d} f (0).
$$
Here $y_1,\ldots,y_d \in \K^n$. We set $y=(y_1,\ldots,y_d) \in \K^{nd}$.

\begin{lemma}\label{lem:blrnprime_deriv}
Under the conditions of Theorem~\ref{thm:blrnprime}, there exist $g_1,\ldots g_c \in \P_{d-1}(\K^{nd})$, $c=c^{(\ref{lem:blrnprime_deriv})}(d,s)$, and $\Gamma:\K^c \rightarrow \K$, such that $\Delta f(y)=\Gamma(g_1(y)),\ldots, g_c(y))$.
\end{lemma}

\begin{proof}
We use two basic facts about $\Delta f(y)$. First, it is a multilinear polynomial in each of $y_1,\ldots,y_d$. That is, for any $a_1,\ldots,a_d \in \K$,
$$
\Delta f(a_1 y_1 ,\ldots, a_d y_d) = a_1 \cdots a_d \Delta f(y_1,\ldots,y_d).
$$
Second, it is biased. By the Gowers-Cauchy-Schwarz inequality (see eg \cite{GT09}), and the fact that derivatives and traces commute, we have that
$$
|\E_{y \in \K^{nd}}[e(\Tr(\Delta f(y)))]| = \|e(\Tr(f))\|_{U^d}^{2^d} \ge \|e(\Tr(f))\|_{U^1}^{2^d} = |\E_{x \in \K^n}[e(\Tr(f(x)))]|^{2^d} \ge |\K|^{-s 2^d}.
$$
Thus, we can apply Lemma~\ref{lem:blrnprime_trace} to $\Delta f$ (with $s'=s 2^d$) and obtain that there exist polynomials $g_1,\ldots,g_c \in \P_d(\K^{nd})$ and a function
$\Gamma:\F_p^c \to \F_p$, where $c=c^{(\ref{lem:blrnprime_trace})}(d,s')$, such that
$$
\Tr(\Delta f(y))=\Gamma'(\Tr(g_1(y)),\ldots, \Tr(g_c(y))).
$$
Each of $g_i$ is a derivative of $\Delta f$. Since $\Delta f$ is multilinear in $y_1,\ldots,y_d$, each $g_i$ can be decomposed as the sum of $d$ terms, each is
a multilinear polynomial in $d-1$ of the $d$ sets of variables $y_1,\ldots,y_d$. With the price of increasing $c$ to at most $dc$,
we assume for simplicity that each $g_i$ is multilinear in $\{y_j: j \in [d] \setminus \ell_i\}$ for some $\ell_i \in [d]$.

Fix $a \in \K$. Define $y^a = (a y_1,y_2,\ldots,y_d) \in \K^{nd}$. We have that $\Delta f(y^a) = a \Delta f(y)$ since $f$ is multilinear. We also have $g_i(y^a) = a_i g_i(y)$
where $a_i=1$ if $\ell_i=1$ and $a_i=a$ otherwise. Thus
$$
\Tr(a \Delta f(y))=\Gamma'(\Tr(a_1 g_1(y)),\ldots, \Tr(a_c g_c(y))).
$$
To conclude, note that any $z \in \K$ is uniquely determined by $(\Tr(a z): a \in \K)$. Thus, if we know the value of $g_1(y),\ldots,g_c(y)$, we can compute $\Tr(a g_i(y))$ for
all $a \in \K, i \in [c]$, from that compute $\Tr(a \Delta f(y))$ for all $a \in \K$, and hence compute $\Delta f(y)$. We thus obtain that
$$
\Delta f(y)=\Gamma(g_1(y),\ldots,g_c(y))
$$
where $\Gamma:\K^c \to \K$ is as described above.
\end{proof}

To conclude, we relate the rank of $f$ to the rank of $\Delta f$. This is the only place in the proof where we use the assumption that $p>d$.
Assume that $\Delta f(y)=\Gamma(g_1(y),\ldots,g_c(y))$ where $g_1,\ldots,g_c \in \P_d(\K^{nd})$.
Let $f_d$ denote the homogeneous part of degree $d$ of $f$, and let $f_{<d} = f-f_d$. We have that
$$
\Delta f (x,\ldots,x) = d! \cdot f_d(x).
$$
If $d<p$ then $d!$ is invertible in $\K$, and hence
$$
f(x) = \frac{1}{d!}\Delta f(x,\ldots,x) + f_{<d}(x)=\frac{1}{d!} \Gamma(g_1(x,\ldots,x),\ldots,g_c(x,\ldots,x))+f_{<d}(x).
$$
Thus, $f$ has rank at most $c+1$. This concludes the proof of Theorem~\ref{thm:blrnprime}.

\section{Application: List decoding Reed-Muller codes over large fields}\label{sec:list}

\subsection{Notation and Preliminaries}

Let $\F$ be a prime finite field. A code $\CC \subset \{\F^n \to \F\}$ is a subset of functions from $\F^n$ to $\F$, where functions in the code are called codewords. The distance between
two functions $f,g:\F^n \to \F$ is the fraction of coordinates where they disagree,
$$
\dist(f,g) := \frac{1}{|\F|^n} \l|\{x \in \F^n: f(x) \ne g(x)\}\r|.
$$
The minimum distance of a code $\CC$ is
$$
\dist_{min}(\CC):=\min_{f\neq g \in \CC} \{\dist(f,g)\}.
$$
A code $\CC$  is linear if it is a linear subspace over $\F$. For a linear code, $\dist_{min}(\CC)=\min_{0 \ne f \in \CC} \{\dist(f,0)\}$.
For a code $\CC$ and a function $g:\F^n \to \F$, the set of codewords at distance at most $\rho$ from $g$ is denoted by
$$
B_{\CC}(g,\rho) := \{f \in \CC: \dist(f,g) \le \rho\}.
$$
The list decoding size of $\CC$ at radius $\rho$ is the maximal number of codewords at distance $\rho$ from any possible function,
$$
L_{\CC}(\rho) := \max_{g:\F^n \to \F} |B_{\CC}(g,\rho)|.
$$

The Reed-Muller code $\RM_{\F}(n,d)$ is the evaluations of all polynomials $f:\F^n \to \F$ of degree at most $d$. So using our previous notation, $\RM_{\F}(n,d) = \P_d(\F^n)$.
As we assume $d<|\F|$, its minimal distance is given by
$$
\dist_{\min}(\RM_{\F}(n,d)) = \min\l\{\Pr_{x \in \F^n}[f(x) \ne 0]: f:\F^n \to \F, f \ne 0, \deg(f) \le d\r\} = 1-\frac{d}{|\F|}.
$$
The main theorem we prove is that Reed-Muller codes, for constant degrees, are list decodable up to their minimal distance. We also extend this to estimate the number of codewords
in balls of larger radii.

\noindent  \textbf{Theorem~\ref{THM:listdecode}.} Let $d,s \in \N$. There exists $c=c(d,s)$ such that the following holds.
For any prime finite field $\F$ with $|\F|>d$ and any $n \in \N$,
$$
L_{\RM_{\F}(n,d)}\l(1-\frac{d}{|\F|}-\frac{1}{|\F|^s}\r) \le |\F|^c.
$$
Moreover, for any $1 \le e<d$,
$$
L_{\RM_{\F}(n,d)}\l(1-\frac{e}{|\F|}-\frac{1}{|\F|^s}\r) \le |\F|^{c \cdot n^{d-e}}.
$$
Both bounds are tight, up to the exact value of $c=c(d,s)$. Thus, this resolves the weight distribution problem asymptotically for all fields, first raised by MacWilliams and Sloane in 1977 \cite{MS77}.

 The proof will follow from a series of propositions which we state next.

Let $\RM_{\F}(n,d,k)$ be a subcode of $\RM_{\F}(n,d)$, which consists of polynomials of degree $\le d$ and rank $\le k$. We first reduce
the problem of list decoding Reed-Muller codes to list decoding a low rank subcode.

\begin{lemma}\label{lem:johnsonred}
Let $e \leq d,s \in \N$. There is $k=k(d,s)$ such that for any prime field $\F$ with $|\F|>d$ and any $n \in \N$,
$$
L_{\RM_{\F}(n,d)}\l(1-\frac{e}{|\F|}-\frac{1}{|\F|^s}\r) \leq |\F|^{2s} \cdot L_{\RM_{\F}(n,d,k)}\l(1-\frac{e}{|\F|}-\frac{1}{|\F|^s}\r).
$$
\end{lemma}

So, from now on we restrict our attention to $\RM_{\F}(n,d,k)$. Recall that $\Delta(\F)$ is the probability simplex over $\F$, that we naturally embed $\F \subset \Delta(\F)$.
For $g:\F^n \to \F$ let $p(g):\F^n \to \Delta(\F)$ be this embedding extended to functions. With this notation, for $f,g:\F^n \to \F$ we have $\dist(f,g) = 1-\ip{p(f),p(g)}$.
So, from now on we extend our study to functions $\varphi:\F^n \to \Delta(\F)$, which can be viewed as randomized functions. The definition of the codewords in $\C$ which are close to a function can be extended to randomized functions following the above discussion:
$$
B_{\CC}(\varphi,\rho) = \{f \in \CC: \ip{p(f),\varphi} \ge 1-\rho\}.
$$
Let $\cF=\{h_1,\ldots,h_c: \F^n \to \F\}$. We say that $\varphi$ is $\cF$-measureable if $\varphi=\Gamma(\cF)$ for some function $\Gamma:\F^{|\cF|} \to \Delta(\F)$.
Recall that $\E[\varphi | \cF]:\F^n \to \Delta(\F)$ as the average of $\varphi$ with respect to $\cF$,
$$
\E[\varphi | \cF](x) = \E\l[\varphi(y): y \in \F^n, \cF(x)=\cF(y)\r].
$$
Clearly, $\E[\varphi|\cF]$ is $\cF$-measurable. Moreover, for any $\xi:\F^n \to \Delta(\F)$
which is $\cF$-measurable, we have
$$
\ip{\xi, \varphi} = \ip{\xi, E[\varphi | \cF]}.
$$
We next show that the list decoding problem for low rank codes can be further reduced to the case where the center $g$ is a measurable with respect to a small polynomial factor
of bounded degree. More accurately, it can be list decoded to this latter problem.

\begin{lemma}\label{lem:proxy}
Fix $d,k,s \in \N$. There exist $c=c(d,k,s) \in \N$ such that the following holds. Let $\F$ be a prime field with $|\F|>d$ and let $n \in \N$.
For any $\varphi:\F^n \to \Delta(\F)$ there exists a family of $|\F|^{c}$ sets of polynomials $\cF_i \subset \RM_{\F}(n,d-1)$, $1 \le i \le |\F|^{c}$,
of size $|\cF_i| \le c$ each, such that
$$
\forall f \in \RM_{\F}(n,d,k) \; \exists 1 \le i \le |\F|^{c}, \left| \ip{p(f),\varphi} -  \ip{p(f),\E[\varphi | \cF_i]} \right| \le |\F|^{-s}.
$$
\end{lemma}

As a corollary, we bound the list decoding size in $\RM_{\F}(n,d,k)$ by the list decoding size when the centers are measurable functions for a system of a few polynomials.

\begin{cor}\label{cor:proxy}
Let $\CC=\RM_{\F}(n,d,k)$. Then for any $0 \le \rho \le 1$,
$$
B_{\CC}(\varphi, \rho) \subset \bigcup_{1 \le i \le |\F|^{c}} B_{\CC}(\E[\varphi|\cF_i], \rho+|\F|^{-s}).
$$
\end{cor}

Finally, we prove bounds for the list decoding problem for low rank codes, where the center is measurable with respect to a polynomial factor. In
fact, we can even ignore the restriction that the code is low rank, as the restriction on the center is sufficient to obtain the bounds.

\begin{lemma}\label{lem:final}
Fix $d,s,c \in \N$. There exists $c'=c'(d,s,c)$ such that the following holds. Let $\F$ be a prime field with $|\F|>d$ and let $n \in \N$.
Let $\calF \subset \RM_{\F}(n,d-1)$ of size $|\calF| \le c$, and let $\varphi:\F^n \to \Delta(\F)$ be $\calF$-measurable.
Then
$$
\l|B_{\RM_{\F}(n,d)}\l(\varphi, 1-\frac{e}{|\F|}-\frac{1}{|\F|^s}\r)\r| \le |\F|^{c' \cdot n^{d-e}}.
$$
In particular, for any $k \in \N$,
$$
\l|B_{\RM_{\F}(n,d,k)}\l(\varphi, 1-\frac{e}{|\F|}-\frac{1}{|\F|^s}\r)\r| \le |\F|^{c' \cdot n^{d-e}}.
$$
\end{lemma}

With the above in place, we are ready to prove our main theorem of the section.
\begin{proof}[Proof of Theorem~\ref{THM:listdecode}]
Let $\rho:=1-\frac{e}{|\F|}-\frac{1}{|\F|^s}$. By Lemma~\ref{lem:johnsonred} there is $k:=k(d,s)$ such that
$$
L_{\RM_{\F}(n,d)}(\rho) \le |\F|^{2s}\cdot L_{\RM_{\F}(n,d,k)}\l(\rho+\frac{1}{|\F|^{2s}}\r).
$$
Let $g:\F^n \to \F, \varphi=p(g)$. Let $\CC=\RM_{\F}(n,d,k)$. Then, by Corollary~\ref{cor:proxy}, for some $c=c(d,s,k)$ we have
$$
\l|B_{\CC}(\varphi, \rho)\r| \le  \sum_{i=1}^{|\F|^{c}} \l|B_{\CC}\l(\E[\varphi|\cF_i], 1-\frac{e}{|\F|}-\frac{1}{|\F|^{2s}}\r)\r|,
$$
where each $\cF_i \subset \RM_{\F}(n,d-1)$ of size $|\cF_i| \le c$.
Finally, by Lemma~\ref{lem:final}, for some $c'=c'(d,s,c)$, we have that for every $1 \le i \le |\F|^{c}$,
$$
\l|B_{\CC}\l(\E[\varphi|\cF_i], 1-\frac{e}{|\F|}-\frac{1}{|\F|^{2s}}\r)\r| \le |\F|^{c' n^{d-e}}.
$$
We conclude that
$$
L_{\RM_{\F}(n,d)} \l (1 - \frac{e}{|\F|} - \frac{1}{|\F|^s} \r) \le |\F|^{2s+c+c' n^{d-e}}.
$$
\end{proof}

We prove Lemma~\ref{lem:johnsonred}, Lemma~\ref{lem:proxy} and Lemma~\ref{lem:final} in the following subsections.

\subsection{Proof of Lemma~\ref{lem:johnsonred}}

We state the Johnson bound first, which provides bounds on the list decoding size for any code, based just on the minimal distance of the code~\cite{johnson}.

\begin{lemma}[Johnson bound]\label{lem:johnson}Let $\CC \subseteq \{\F^n \to \F\}$. Suppose that $\dist_{min}(\CC) \geq 1-\frac{1}{|\F|}-\eps$. Then, $$L_{\CC}\l(1-\frac{1}{|\F|}-\sqrt{\eps}\r) \leq 1/\eps^2.$$
\end{lemma}

\begin{proof}[Proof of Lemma~\ref{lem:johnsonred}]
Set $k=k(d,s)=c^{(\ref{thm:blr})}(d,2s)$. Fix arbitrary $g:\F^n \to \F$. Let
$$
L=\l\{f \in \RM_{\F}(n,d): \dist(f,g) \leq 1-\frac{e}{|\F|}-\frac{1}{|\F|^s}\r\}.
$$
Let $m=|L|$ and $L=\{f_1,\ldots , f_m\}$. Construct a graph $G=(L,E)$ where $(f_i,f_j) \in E$ if $\rank(f_i-f_j) \leq k$. Let $I \subseteq L$ be a maximal independent set.

\begin{claim}$\dist_{min}(I) \geq 1-\frac{1}{|\F|}-\frac{1}{|\F|^{2s}}$.
\end{claim}
\begin{proof}Let $f=f_i-f_j \neq 0$ for $f_i,f_j \in I$. Since $\rank(f) > k(d,s)=c^{(\ref{thm:blr})}(d,2s)$, and therefore, $\rank(a \cdot f)>k$ for all $a \in \F, a \neq 0$, we have by Theorem~\ref{thm:blr} that $\E\l[e\l(a \cdot f(x)\r)\r] \leq |\F|^{-2s}$. Thus,
$$
1-\dist(f_i,f_j)=\Pr_{x \in \F^n}[f(x)=0] = \frac{1}{|\F|} \sum_{a \in \F}\E\l[e\l(a \cdot f(x)\r)\r] \le \frac{1}{|\F|} + \frac{1}{|\F|^{2s}}.
$$
\end{proof}
By the above claim, using the Johnson bound on $I$, we have that
\begin{equation}\label{eq:johnson1}
L_{I}\l(1-\frac{1}{|\F|}-\frac{1}{|\F|^{s}}\r) \leq |\F|^{2s}.
\end{equation}

Next, consider any $f \in I$. Say $h_1,\ldots , h_D \in \RM_{\F}(n,d,k)$ are such that $(f+h_i,f) \in E$. As $\dist(g,f+h_i)\le1-\frac{e}{|\F|}-\frac{1}{|\F|^s}$, we have that $\dist(g-f,h_i)\le 1-\frac{e}{|\F|}-\frac{1}{|\F|^s}.$ Thus,
\begin{equation}\label{eq:johnson2}
D \leq L_{RM_{\F}(n,d,k)}\l(1-\frac{e}{|\F|}-\frac{1}{|\F|^s}\r).
\end{equation}
Combining Equation~\eqref{eq:johnson1} and Equation~\eqref{eq:johnson2} we conclude that $$L_{\RM_{\F}(n,d)}\l(1-\frac{e}{|\F|}-\frac{1}{|\F|^s}\r) \leq |\F|^{2s} \cdot L_{\RM_{\F}(n,d,k)}\l(1-\frac{e}{|\F|}-\frac{1}{|\F|^s}\r).
$$
\end{proof}

\subsection{Proof of Lemma~\ref{lem:proxy}}
The proof of Lemma~\ref{lem:proxy} requires several refinements of $\RM_{\F}(n,d,k)$. First,
for $\cF \subset \RM_{\F}(n,d-1)$ a family of polynomials of degree $\le d-1$,
define $\RM_{\F}(n,d,k,\cF)$ to be the family of degree $d$ polynomials, which can be decomposed as a function of
the polynomials in $\cF$, and $k$ additional polynomials of degree $\le d-1$.

For $\bk=(k_1,\ldots,k_{d-1}) \in \N^{d-1}$ let $|\bk|=\sum k_i$.
The code $\RM_{\F}(n,d,\bk,\cF)$ is a subcode of $\RM_{\F}(n,d,|\bk|,\cF)$, defined as
family of degree $d$ polynomials, which can be decomposed as a function of the polynomials in $\cF$, and $|\bk|$ additional polynomials, with $k_i$ polynomials of degree $i$,
for $1 \le i \le d-1$. The following statement of the theorem allows for a streamlined inductive proof.

\begin{thm}\label{thm:streamlined}
Fix $d,s \in \N, \bk \in \N^{d-1},\cF \subset \RM_{\F}(n,d-1)$ and let $\CC=\RM_{\F}(n,d,\bk,\cF)$. There exist $c=c(d,\bk,s,|\cF|) \in \N$ such that the following holds.
For any $\varphi:\F^n \to \Delta(\F)$ there exists a family of $|\F|^c$ sets of polynomials $\cF_i \subset \RM_{\F}(n,d-1)$, $1 \le i \le |\F|^c$,
of size $|\cF_i| \le c$ each, such that
$$
\forall f \in \CC \; \exists 1 \le i \le |\F|^c, \left| \ip{p(f),\varphi} -  \ip{p(f),\E[\varphi | \cF_i]} \right| \le |\F|^{-s}.
$$
\end{thm}

\paragraph{The simplex.}

Recall that for $f:\F^n \to \F$ we have $p(f):\F^n \to \Delta(\F)$. Define $q(f):=p(f)-\frac{1}{|\F|}$, so that $\sum_{y \in \F} q(f)(x)_y=0$ for all $x \in \F^n$.
For $a \in \F^n, b \in \F$, define $\ell_{a,b}:\F^{n} \to \F$ by $\ell_{a,b}(x)=\ip{a,x}+b$. We prove the following analogue of Fourier expansion over the simplex. We will refer to it as the Fourier simplex decomposition.

\begin{lemma}\label{lem:fourier}
Let $g:\F^n \to \F$. Then,
$$
q(g)(x)=\sum_{a \in \F^n, 0 \ne b \in \F} \alpha_{a,b}q(\ell_{a,b})(x),
$$ where
$$
\alpha_{a,b}=\langle q(g),q(\ell_{a,b})\rangle - \langle q(g),q(\ell_{a,0})\rangle
$$
are unique and satisfy $\alpha_{a,b} \in [-1,1]$. We denote the $\alpha_{a,b}$ by $\widehat{q(g)}(a,b)$.
\end{lemma}
\begin{proof}
We first construct a basis for the subspace $V \subseteq (\R^{|\F|})^{|\F|^n}$ defined as follows. We index the coordinates of $v \in (\R^{|\F|})^{|\F|^n}$
as $v_{x,y}$, with $x \in \F^n, y \in \F$. Then
$$
V=\l\{v_{x,y}  \in \R^{|\F|^{n+1}}: \sum_{y \in \F} v_{x,y}=0 \; \forall x \in \F^n\r\}.
$$
Note that $q(g) \in V$ for any $g:\F^n \to \F$. Also, $\dim(V)=|\F|^n(|\F|-1)$.
We next establish that the set of vectors
$$
I=\l\{q(\ell_{a,b}):a \in \F^n, 0 \neq b \in \F \r\} \subseteq V
$$
is a basis for $V$. First, note that $|I|=|\F|^n(|\F|-1)$. To prove linear independence of $I$, suppose that
\begin{equation}\label{eq:li}
\Lambda:=\sum_{a \in \F^n,0 \ne b \in \F}\alpha_{a,b}q(\ell_{a,b})=0.
\end{equation}
Let $V_a=\spana\{q(\ell_{a,b}): 0 \ne b \in \F\}.$ Then, by Equation~(\ref{eq:li}), $\sum_a v_a=0$ where $v_a = \sum_{0 \ne b \in \F}\alpha_{a,b}q(\ell_{a,b}) \in V_a$. We now note that $\langle v_a, v_{a'} \rangle =0$ if $a \neq a'$. Indeed,
$$
\langle v_a,v_{a'} \rangle=\l\langle \sum_{b \neq 0} \alpha_{a,b}q(\ell_{a,b}),\sum_{b' \neq 0}\alpha_{a',b'}q(\ell_{a',b'}) \r\rangle=0,
$$ since for any $a \ne a' \in \F^n$ and any $b,b' \in \F$,
$$
\langle q(\ell_{a,b}), q(\ell_{a',b'})\rangle=\Pr_{x \in \F^n}\l[\ip{a,x}+b = \ip{a',x}+b'\r] - \frac{1}{|\F|} = 0.
$$
In particular, $\langle v_a, v_a \rangle =0$ for all $a \in \F^n$ which implies that $v_a=0$. Fix an arbitrary $a \in \F^n$. We now show that $v_a=\sum_{b \neq 0} \alpha_{a,b}q(\ell_{a,b})=0$ implies that $\alpha_{a,b}=0$ for all $b \neq 0$. Indeed, fix $x \in \F^n$ such that $\ip{a,x}=0$. Then $v_a(x)_y = \alpha_{a,y} - \sum_{b \ne 0} \alpha_{a,b}$ if $y \ne 0$, and $v_a(x)_0 = - \sum_{b \ne 0} \alpha_{a,b}$. As we have that $v_a(x)_y = 0$ for all $y \in \F$, it must be that $\alpha_{a,b}=0$ for all $0 \ne b \in \F$.

Thus, $I$ indeed forms a basis for $V$. The uniqueness of the $\alpha_{a,b}$ follows from the linear independence of $I$.
So far, we have established that
\begin{equation}\label{eq:basis}
q(g)(x)=\sum_{a \in \F^n,0 \ne b \in \F}\alpha_{a,b}q(\ell_{a,b})(x).
\end{equation}
Using the simple fact that $$\langle q(\ell_{a,b}), q(\ell_{a',b'}) \rangle=\Pr[\ip{a,x}+b=\ip{a',x}+b']-\frac{1}{|\F|},$$  we record the following observation.
$$\langle q(\ell_{a,b}), q(\ell_{a',b'}) \rangle = \bigg \{
\begin{array}{ll}
0&\textrm{if } a \neq a' \\
1-\frac{1}{|\F|}&\textrm{if } a=a', b=b' \\
-\frac{1}{|\F|} &\textrm{if } a=a', b \ne b'
\end{array}
$$
Taking inner product on both sides of Equation~\eqref{eq:basis} with $q(\ell_{a,b})$ we get,
$$
\langle q(g),q(\ell_{a,b}) \rangle=\l(1-\frac{1}{|\F|}\r)\alpha_{a,b} -\frac{1}{|\F|}\l(\sum_{b' \neq 0, b}\alpha_{a,b'}\r)=\alpha_{a,b} - \frac{1}{|\F|} \sum_{b' \neq 0}\alpha_{a,b'}.
$$
Summing for all $b \neq 0$, we obtain that
\begin{equation}\label{eq:alpha}
\sum_{b \neq 0}\langle  q(g),q(\ell_{a,b}) \rangle = \frac{1}{|\F|} \sum_{b' \neq 0}\alpha_{a,b'}.
\end{equation}
Thus,
\begin{equation}\label{eq:alpha1}
\alpha_{a,b}=\langle q(g),q(\ell_{a,b}) \rangle + \sum_{b' \neq 0}\langle  q(g),q(\ell_{a,b'}) \rangle.
\end{equation}
Next, we observe that $\sum_{b \in \F} q(\ell_{a,b})=0$. This is since
$$
\sum_{b \in \F} q(\ell_{a,b})_{x,y} = \sum_{b \in \F} \l( \Pr[\ip{a,x}+b=y] - \frac{1}{|\F|}\r) = 1-1=0.
$$
So we have
$$
\alpha_{a,b}=\langle q(g),q(\ell_{a,b}) \rangle - \langle q(g),q(\ell_{a,0}) \rangle.
$$
Since $\langle q(g),q(\ell_{a,b}) \rangle \in [-\frac{1}{|\F|},1-\frac{1}{|\F|}]$ for all $b \in \F$, we obtain that $\alpha_{a,b} \in [-1,1]$.
This finishes the proof
\end{proof}

\paragraph{Weak regularity on the simplex.}

We prove the following lemma. In the following, $X,Y$ are arbitrary finite sets, where we will later apply the lemma to $X=\F^n, Y=\F$. The proof is similar to Frieze-Kannan weak regularity~\cite{FK99} but generalized to the simplex.
\begin{lemma}\label{lem:weakreg}Let $\eps>0$ be arbitrary. Let $\varphi:X \to \Delta(Y)$ be arbitrary. Let $\calF$ be a collection of functions $f:X \to Y$. Then, there exist $f_1,\ldots , f_c \in \calF$, $c \leq 1/\eps^2$ such that
$$
\varphi=\frac{1}{|Y|}+\sum_{i=1}^c \alpha_i q(f_i)+h,
$$ where $|\alpha_i| \leq 1$ and $h$ satisfies that for all $f \in \calF$, $$\l|\langle h,q(f)\rangle \r| \leq \eps.$$
\end{lemma}
\begin{proof}
Let $\varphi'=\varphi-\frac{1}{|Y|}$. We will define a sequence of functions $\varphi_i \in \F^n \to \R^{\F}$. Initialize $\varphi_0:=0$.
Given $\varphi_i$, if there exists $f_i \in \calF$ such that $\l|\langle \varphi'-\varphi_i,q(f_i)\rangle\r| = \alpha_i > \eps$, set $\varphi_{i+1}:=\varphi_i+\alpha_{i+1}q(f_i)$.
We show that the process terminates quickly. To that end, define $\delta_i:=||\varphi'-\varphi_i||_2^2$. Then
\begin{eqnarray*}
\delta_{i+1}&=&||\varphi'-\varphi_i-\alpha_i q(f_i)||_2^2\\
&=&\delta_i +\alpha_i^2||q(f_i)||_2^2-2\langle \varphi'-\varphi_i,\alpha_i q(f_i)\rangle \\
&=&\delta_i+\alpha_i^2 (1-1/|Y|)-2\alpha_i^2 \\
& \leq & \delta_i -\alpha_i^2\\
& \leq & \delta_i -\eps^2.
\end{eqnarray*}

Additionally, $\delta_0=||\varphi'||_2^2 \leq 1$ and  $\delta_i \geq 0$ for all $i$. Thus, the process terminates after $\le 1/\eps^2$ steps.
At the end of the process, we have
$$\varphi'=\sum_{i=1}^c \alpha_i q(f_i)+h,$$ where $h$ satisfies that for all $f \in \calF$, $\l|\langle h,q(f)\rangle \r| \leq \eps$ and  $|\alpha_i| \le \sqrt{\delta_i - \delta_{i+1}} \le 1$.
\end{proof}

\begin{proof}[Proof of Theorem~\ref{thm:streamlined}]
The proof is by induction on $d,s,\bk,|\cF|$. For $\bk$, we use the lexicographic order on $\N^{d-1}$ which is well founded to define a Noetherian induction. Let $\CC=\RM_{\F}(n,d,\bk,\cF)$, and fix $\varphi:\F^n \to \Delta(\F)$ and $s \ge 1$, and set $e=|\cF|$.

We first argue that we may assume that $\cF$ is regular. For a rank function $R_1:\N \to \N$ to be determined later (as a function of $d,\bk,s$), regularize $\cF$ to obtain an $R_1$-regular factor $\cF'$. Note that
$$
\RM_{\F}(n,d,\bk,\cF) \subset \RM_{\F}(n,d,\bk,\cF').
$$
Thus, we may instead study $\RM_{\F}(n,d,\bk,\cF')$. So, we simply assume from now on that $\cF$ is $r_1$-regular for some $r_1=R_1(d,\bk,s,|\cF|)$ to be determined later.

Let $f \in \CC$. By definition, we can decompose $f$ as
$$
f=\Gamma(\cH, \cF),
$$
where $\cH=\{h_1,\ldots,h_k\}$ is a family of $k$ polynomials of degree $\le d-1$ and where $\Gamma:\F^{k+e} \to \F$ is some function.
We argue that we can also assume that $\cH \cup \cF$ is regular. If $\cH \cup \cF$ is not $r_2=r_1-1$ regular, then
$$
\rank\left(\sum a_i h_i + \sum b_i f_i\right) \le r_2,
$$
for some $a_i,b_i \in \F$, not all zero. Let $d'$ be the maximal degree of a polynomial appearing in the linear combination with a nonzero coefficient. It cannot be that all these
polynomials are in $\{f_i\}$, as we assumed that the rank of $\cF$ is at least $r_1$. So, $a_i \ne 0$ for some $i$ where $\deg(h_i)=d'$. This means $h_i$ can be expressed
as a function of the other polynomials in $\cH \cup \cF$, and an additional set $\cH'$ of $r_1-2$ polynomials of degrees $\le d'-1$. So, if we construct $\bk'$ from $\bk$ by
reducing the number of polynomials of degree $d'$ by one, and increasing the number of polynomials of degrees $\le d'-1$ by $r_2$, then in fact we have
$$
f \in \RM_{\F}(n,d,\bk',\cF).
$$
Thus, we may apply the theorem by induction in order to handle these polynomials, since $\bk'<\bk$ in the lexicographic order. So, we assume from now on that $\cF \cup \cH$ is $r_2$-regular.

Let $\psi = \E[\varphi|\cF]$. We will include $\cF$ as one of our sets $\cF_i$, and hence handle any $f$ for which $|\ip{q(f),\psi}-\ip{q(f),\varphi}| \le |\F|^{-s}$. So, from now on
we consider only $f$ for which $|\ip{q(f),\psi} - \ip{q(f),\varphi}| \ge |\F|^{-s}$. Decomposing $\Gamma$ to its Fourier simplex decomposition (Lemma~\ref{lem:fourier}), and applying this to decompose $f$, we obtain that
$$
q(f)(x) = \sum_{a \in \F^{k}, b \in \F^e, 0 \ne c \in \F} \widehat{\Gamma}(a,b,c) \cdot q\left(\sum a_i h_i(x) + \sum b_i f_i(x)+c\right),
$$
where $|\widehat{\Gamma}(a,b,c)| \le 1$. Note that whenever $a=0$, we have
$$
\ip{q\left(\sum b_i f_i+c\right), \varphi} = \ip{q\left(\sum b_i f_i+c\right), \psi},
$$
since $\sum b_i f_i+c$ is $\cF$-measurable. Hence, there must exist $0 \ne a \in \F^{k}, b \in \F^e, c \ne 0$, such that
$$
\left| \ip{q\left(\sum a_i h_i(x) + \sum b_i f_i(x)+c\right), \varphi-\psi} \right| \ge |\F|^{-(s+k+e+1)}.
$$
As $\E[\varphi-\psi]=0$ we equivalently have
$$
\left| \ip{p\left(\sum a_i h_i(x) + \sum b_i f_i(x)+c\right), \varphi-\psi} \right| \ge |\F|^{-(s+k+e+1)}.
$$
Next, we decompose by Lemma~\ref{lem:weakreg} both $\varphi$ and $\psi$, and
subtract the decompositions obtain that
$$
\varphi-\psi = \sum_{t=1}^{\ell} \gamma_t \cdot  q(w_t) + \xi,
$$
where $\gamma_t \in [-1,1]$, $w_t \in \CC$, $\xi:\F^n \to \R^{\F}$ satisfies that $|\ip{\xi,f}| \le |\F|^{-2(s+k+e+1)}$ for all $f \in \CC$, and $\ell \le |\F|^{4(s+k+e+1)}$.
There must exist $t \in [\ell]$ such that
$$
\left| \ip{p\left(\sum a_i h_i(x) + \sum b_i f_i(x)+c\right), q(w_t)}\right| \ge |\F|^{-5(s+k+e)}.
$$
As $w_t \in \CC$ we can decompose it as as a function of $k$ polynomials of degree $\le d-1$ and $\cF$. Let $R_3:\N \to \N$ be large enough
to be determined later (as a function of $d,\bk,s$). We regularize these polynomials to be $R_3$-regular, and obtain a collection of $\le c_1(d,\bk,|\cF|,s)$ polynomials.
We choose $R_1$ large enough so that $R_1(e) > c_1(d,\bk,e,s)$. This ensures that $\cF$ does not change in the regularization process. Hence we have
$$
w_t = \Gamma_t(\cH_t \cup \cF),
$$
where $\cH_t \cup \cF$ is $R_3$-regular, $|\cH_t|=k_t \le c_1(d,\bk,e,s)$, $\cH_t=\{h_{t,1},\ldots,h_{t,k_t}\}$ and $\Gamma_t:\F^{k_t+e} \to \F$ is some function.
Decomposing $\Gamma_t$ to its Fourier decomposition, and applying this to decompose $w_t$, we obtain that
$$
q(w_t)(x) = \sum_{a \in \F^{k_t}, b \in \F^e, c \ne 0} \widehat{\Gamma}_t(a,b,c) \cdot q\left(\sum a_i h_{t,i}(x) + \sum b_i f_i(x)+c\right).
$$
So, there must exist $a' \in \F^{k_t}, b' \in \F^e, c' \ne 0$ such that
$$
\left| \ip{p\left(\sum a_i h_i(x) + \sum b_i f_i(x)+c\right), q\left(\sum a'_i h_{t,i}(x) + \sum b'_i f_i(x)+c'\right)}  \right| \ge |\F|^{-5(s+k+e+1)-(k_t+e+1)},
$$
which equivalently means that, for $b''_i = b_i-b'_i$ and $c''=c-c'$, that
$$
\left| \Pr_{x \in \F^n} \left[\sum a_i h_i(x) - \sum a'_i h_{t,i}(x) + \sum b''_i f_i(x) + c'' = 0 \right] - \frac{1}{|\F|} \right| \ge |\F|^{-(5(s+k+e+1)+(k_t+e+1))}.
$$
This implies that
$$
\rank\left(\sum a_i h_i(x) - \sum a'_i h_{t,i}(x) + \sum b''_i f_i(x)\right) \le r_4=r_4(d,k,e,s).
$$
Let $d'$ be the maximal degree of a polynomial appearing in the linear combination with a nonzero coefficient. By choosing $R_3$ large enough,
we guarantee that it cannot be the case that all the polynomials of degree $d'$ are in $\cH_t \cup \cF$. So, $a_i \ne 0$ for some $i$ such that $\deg(h_i) = d'$. So, we can express $h_i$ as a function of all the other polynomials
in $\cH \cup \cH_t \cup \cF$, and an additional set of $r_4$ polynomials of degree $\le d'-1$. Thus, we define $\cF_t = \cF \cup \cH_t$, and construct $\bk'$ from $\bk$ by decreasing the number of polynomials
of degree $d'$ by one, and increase the number of polynomials of any lower degree by $r_4$, then $\bk' < \bk$ and we obtain that in fact
$$
f \in \RM_{\F}(n,d,\bk',\cF_t).
$$
Crucially, the sets $\cF_t$ were obtained depending only on $\varphi$ and $\cF$. Thus, we can apply the theorem by induction to each of them.
Let $\{\cF_{t,i}: 1 \le i \in |\F|^{c'}\}$ be the sets guaranteed by the theorem, where $c' \le c(d,\bk',s,\cF_t)$. We conclude the proof by taking their union, which has size
$\le |\F|^{4(s+k+e)} \cdot |\F|^{c(d,\bk',s,|\cF_t|)}$, which is bounded by $|\F|^c$ for a large enough $c=c(d,\bk,s,|\cF|)$.
\end{proof}

\subsection{Proof of Lemma~\ref{lem:final}}
The proof of Lemma~\ref{lem:final} is similar to the authors' previous work~\cite{BL14}. Let $\calF=\{h_1,\ldots , h_c\}$ a family of polynomials of degree $\le d-1$, and let $\varphi:\F^n \to \Delta(\F)$ be $\cF$-measurable. Let $\rho:=1-\frac{e}{|\F|}-\frac{1}{|\F|^{s}}$. Fix $f \in \CC$ such that
$$
\langle p(f),\varphi\rangle \ge 1-\rho.
$$
For $a \in \F^c$ define
$$
A_a:=\{x \in \F^n:h_1(x)=a_1,\ldots , h_c(x)=a_c\}.
$$
Define $\Gamma_f:\F^c \to \F$ by setting $\Gamma_f(a)$ to be the most common value $f$ attains on $A_{a}$. Then
\begin{align*}
\Pr[f(x)=\Gamma_f(h_1(x),\ldots,h_c(x))] &= \sum_{a \in \F^c}  \Pr[x \in A_a] \cdot \max_{y^* \in \F} \Pr[f(x)=y^* | x \in A_{a}]\\
&\ge \sum_{a \in \F^c} \Pr[x \in A_{a}] \cdot \E\l[\langle p(f),\varphi\rangle  | x \in A_{a}\r]\\
& = \E[\l[\langle p(f),\varphi\rangle\r]\\
& \ge 1-\rho.
\end{align*}

Let $r_1, r_2:\N \rightarrow \N$ be two non decreasing functions to be specified later, and let $C_{r,d}^{(\ref{lem:reg})}$ be as given in Lemma~\ref{lem:reg}. We will require that for all $m \ge 1$, \begin{equation}\label{eq:r1r2}  r_1(m)\geq r_2(C_{r_2,d}^{(\ref{lem:reg})}(m+1))+C_{r_2,d}^{(\ref{lem:reg})}(m+1)+1.
\end{equation}

Let $\B$ be the factor defined by $\calF$. As a first step, we $r_1$-regularize $\calF$ by Lemma~\ref{lem:reg}. This gives an $r_1$-regular factor $\B'$ of degree at most $d$, defined by polynomials $\calF'=\{h_1',\ldots , h_{c'}':\F^n \rightarrow \F\}$, such that $\B' \succeq_{sem} \B$, $c' \leq C_{r_1,d}^{(\ref{lem:reg})}(c)$ and $\rank(\B')  \geq  r_1(c')$. Let $G_f:\F^{c'} \rightarrow \F$ be defined such that
$$
G_f(h_1'(x),\ldots , h_{c'}'(x))=\Gamma_f(h_1(x),\ldots , h_c(x)).
$$
Then
\begin{equation}\label{eq:fool}
\Pr[G_f(h_1'(x),h_2'(x),\ldots , h_{c'}'(x))=f(x)] \geq 1-\rho.
\end{equation}

Appealing again to Lemma~\ref{lem:reg}, we $r_2$-regularize $\B_f:=\B' \cup \{f\}$. We get an $r_2$-regular factor $\B'' \succeq_{syn} \B'$ defined by the collection $\calF''=\{h_1',\ldots , h_{c'}',h''_1,\ldots , h''_{c''}\}\subseteq \RM_{\F}(n,d-1)$. Note that it is a syntactic refinement of $\B'$ as by our choice of $r_1$, $$\rank(\B') \geq r_1(c')  \geq  r_2(C_{r_2,d}^{(\ref{lem:reg})}(c'+1))+C_{r_2,d}^{(\ref{lem:reg})}(c'+1)+1 \geq r_2(|\B''|)+|\B''|+1.$$
We will choose $r_2$ such that for all $m \ge 1$,
\begin{equation}\label{eq:r2atom}
r_2(m) = \max\l(r^{(\ref{lem:atomsize})}(d,2s+m),r^{(\ref{lem:degree})}(m,d,d)\r).
\end{equation}
Since $f$ is measurable with respect to $\B''$, there exists $F:\F^{c'+c''} \to \F$ such that
$$
f(x)=F(h_1'(x),\ldots , h_{c'}'(x), h''_1(x),\ldots , h''_{c''}(x)).
$$

As will see soon, our goal is to analyze the structure of $F$. We next show that we can have each polynomial in the factor have a disjoint set of inputs.
Let $r \in\N$ be large enough to be determined later. Let $n_1 = r \sum_{i=1}^{c'} \deg(h'_i)$ and $n_2 = r \sum_{i=1}^{c''} \deg(h''_i)$.
Define $y \in \F^{n_1}$ indexed as $y_{i,j,k}$, with $i \in [c'], j \in [r], k \in [\deg(h'_i)]$, and define
$z \in \F^{n_2}$ indexed as $z_{i,j,k}$, with $i \in [c''], j \in [r], k \in [\deg(h''_i)]$.
Define new polynomials $\widetilde{h'_i}(y), \widetilde{h''_i}(z)$ as follows:
\begin{align*}
\widetilde{h'_i}(y)=\sum_{j=1}^r \prod_{k=1}^{\deg(h'_i)} y_{i,j,k} \qquad \forall i \in [c'],\\
\widetilde{h''_i}(z)=\sum_{j=1}^r \prod_{k=1}^{\deg(h''_i)} z_{i,j,k} \qquad \forall i \in [c''].
\end{align*}
Note that the polynomials $\{h'_i: i \in [c']\},\{h''_i: i \in [c'']\}$ are defined over disjoint sets of variables, and that $\deg(\widetilde{h'_i})=\deg(h'_i)$ and $\deg(\widetilde{h''_i})=\deg(h''_i)$.
Define new functions $\widetilde{f}:\F^{n_1+n_2} \rightarrow \F$ and $\widetilde{g}:\F^{n_1} \rightarrow \F$ as  follows:
\begin{align*}
&\widetilde{f}(y,z)=F(\widetilde{h'_1}(y),\ldots , \widetilde{h'_{c'}}(y), \widetilde{h''_1}(z), \ldots , \widetilde{h''_{c''}}(z)),\\
&\widetilde{g}(y)=G_f(\widetilde{h'_1}(y), \ldots , \widetilde{h'_{c'}}(y)).
\end{align*}

\begin{claim}\label{claim:tild}For a large enough $r=r(d,c',c'',s)$ it holds that $\deg(\widetilde{f}) \le d$ and
$$
\l|\Pr_{y \in \F^{n_1}, z \in \F^{n_2}}[\widetilde{f}(y,z)=\widetilde{g}(y)] - \Pr_{x \in \F^n}[f(x)=G_f(h_1'(x),h_2'(x),\ldots , h_c'(x))]\r|\leq  \frac{1}{|\F|^{s+1}}.
$$
\end{claim}

\begin{proof}
The bound $\deg(\widetilde{f}) \le \deg(f) \le d$ follows from Lemma~\ref{lem:degree} since $r_2(|\calF''|) \ge r^{(\ref{lem:degree})}_d(|\calF''|)$. To establish the bound on $\Pr[\widetilde{f}=\widetilde{g}]$, for each $a \in \F^{c'+c''}$ let
$$
p_1(a) = \Pr_{x \in \F^n}[(h'_1(x),\ldots,h'_{c'}(x),h''_1(x),\ldots,h''_{c''}(x))=a].
$$
Applying Lemma~\ref{lem:atomsize} and since our choice of $r_2$ satisfies $\rank(\calF'') \ge r^{(\ref{lem:atomsize})}(d,s+2|\calF''|)$, we have that $p_1$ is nearly uniform over $\F^{c'+c''}$,
$$
p_1(a) = \frac{1 \pm |\F|^{-2s}}{|\F|^{c'+c''}}.
$$
Similarly, let
$$
p_2(a) = \Pr_{y \in \F^{n_1},z \in \F^{n_2}} [(\widetilde{h'_1}(y),\ldots,\widetilde{h'_{c'}}(y),\widetilde{h''_1}(z),\ldots, \widetilde{h''_{c''}}(z))=a].
$$
For $r$ large enough, as the polynomials are evaluated on disjoint variables, it also holds that $$p_2(a) = \frac{1 \pm |\F|^{-2s}}{|\F|^{c'+c''}}.$$

For $a \in \F^{c'+c''}$, let $a' \in \F^{c'}$ be the restriction of $a$ to first $c'$ coordinates,  $a'=(a_1,\ldots ,a_{c'})$. Thus
\begin{align*}
\Pr_{y \in \F^{n_1},z \in \F^{n_2}}[\widetilde{f}(y,z)=\widetilde{g}(y)] &=
\sum_{a \in \F^{c'+c''}} p_2(a) 1_{F(a)=G_f(a')} \\
&= \sum_{a \in \F^{c'+c''}} p_1(a) 1_{F(a)=G_f(a')} \pm |\F|^{-2s} \\
&= \Pr_{x \in \F^n}[f(x)=G_f(h_1'(x),h_2'(x),\ldots , h_c'(x))] \pm |\F|^{-2s}.
\end{align*}
\end{proof}

So, we obtain that
$$
\Pr_{y \in \F^{n_1},z \in \F^{n_2}}[\widetilde{f}(y,z)=\widetilde{g}(y)] \ge \Pr_{x \in \F^n} [f(x) = G_f(h'_1(x),\ldots,h'_{c'}(x))] - |\F|^{-2s} \ge \frac{e}{|\F|}+|\F|^{-2s}.
$$
In the remaining part of the proof, we show that $\deg(\widetilde{h_j''}) \le d-e$. Since, $$\rank(\B'') \ge r_2(|\B''|) \ge r^{(\ref{lem:degree})}(|\B''|,d,d),$$ this implies that $\deg(F) \le d$ by Lemma~\ref{lem:complowdegree}.
This immediately proves that the number of $f \in B_{\CC}(\varphi, \rho)$ is bounded by
$$
|B_{\CC}(\varphi, \rho)| \le (\#\textrm{ of } F) (\# \textrm{ of } h''_1,\ldots,h''_{c''}) \le |\F|^{(c'+c'')^d}|\F|^{O\l(c''n^{d-e}\r)}=|\F|^{O_{d,s,c}\l(n^{d-e}\r)}.
$$

To conclude, we prove the following.
\begin{claim}
$\deg(\widetilde{h''_i}) \le d-e$ for all $i \in [c'']$.
\end{claim}
\begin{proof}
To simplify notations, let $h=h''_i$, $n' = r \cdot \deg(h''_i)$, let $w' \in \F^{n'}$ denote the inputs to $\widetilde{h''_i}$, namely $\{z_{i,j,k}: j \in [r], k \in \deg(h''_i)\}$, and let $w'' \in \F^{n_1+n_2-n'}$ denote all the remaining inputs from $y,z$. Let $n''=n_1+n_2-n'$. Then we have
$$
\widetilde{f}(w',w'')=\Gamma'(w'',h(w')), \quad \widetilde{g}(w'')=\Gamma''(w'').
$$
Let $d_0:=\deg(h)$, where our goal is to prove that $d_0 \le d-e$. Note that as $\deg(\widetilde{f}) \le d$ by Claim~\ref{claim:tild}, we must have $\deg(\Gamma') \le d$. Thus we can expand
$$
\Gamma'(w''_1,\ldots , w''_{n''},t)=\sum_{i=0}^{d'} q_i(w''_1,\ldots,w''_{n''}) t^i,
$$
where $d' \le d$ and $q_{d'} \ne 0$. Moreover, by choosing $r > d^2$, we have that $\deg(h^i) = i \cdot \deg(h)$ for any $i \le d$. Thus, we have $\deg(q_i) \le d - i \cdot d_0$.
We have
$$
\Pr[\widetilde{f}(w'',h(w'))=\widetilde{g}(w'')] = \Pr\l[(q_0-\Gamma'')(w'')+\sum_{i=1}^{d'}q_i(w'')  h(w)^i=0\r].
$$
We upper bound this probability as a combination of two terms. Consider any fixing of $w''$. The probability that $q_{d'}(w'')=0$ is bounded by
$$
\Pr[q_{d'}(w'')=0] \le \frac{\deg(q_{d'})}{|\F|} \le \frac{d - d' d_0}{|\F|}.
$$
Otherwise, we have $q_{d'}(w'') \ne 0$. In such a case, by choosing $r$ large enough (as a function of $s$) we have that $\l |\Pr[h(w)=a] - |\F|^{-1} \r| \le |\F|^{-4s}$ for all $a \in \F$; and hence, if we set $\alpha_i = q_i(w'')$ for $1 \le i \le d'$ and $\alpha_0 = q_0(w'')-\Gamma''(w'')$, then
$$
\Pr_{w' \in \F^{n'}} \l[\sum_{i=0}^{d'} \alpha_i h(w')^i=0\r] = \Pr_{\beta \in \F} \l[\sum_{i=0}^{d'} \alpha_i \beta^i=0\r] \pm |\F|^{-4s} \le \frac{d'}{|\F|} + |\F|^{-4s},
$$
where $\beta \in \F$ is a uniform field element. Combining these bounds, we have that
$$
\Pr[\widetilde{f}(w'',h(w'))=\widetilde{g}(w'')] \le \frac{d - d' d_0}{|\F|} + \l( 1- \frac{d - d' d_0}{|\F|} \r) \frac{d'}{|\F|} + |\F|^{-4s}
$$
Recalling that $\Pr[\widetilde{f}(w'',h(w'))=\widetilde{g}(w'')] \ge \frac{e}{|\F|} + |\F|^{-2s}$, we obtain that
$$
\frac{e}{|\F|} < \frac{d - d' d_0}{|\F|} + \frac{d'}{|\F|}= \frac{d - d' (d_0-1)}{|\F|}.
$$
Thus, $d_0-1 < \frac{d-e}{d'} \le d-e$ and hence $d_0 \le d-e$ as claimed.

\end{proof}

\newcommand{\etalchar}[1]{$^{#1}$}

\end{document}